\documentclass[a4paper,twocolumn,11pt,accepted=2024-04-14]{quantumarticle}
\pdfoutput=1
\usepackage[utf8]{inputenc}
\usepackage[english]{babel}
\usepackage[T1]{fontenc}
\usepackage{amsmath}
\usepackage{hyperref}
\usepackage{tikz}
\usepackage{lipsum}

\usepackage{cite}
\usepackage{amsmath,amsfonts}
\usepackage{amstext} 
\usepackage{array}   
\usepackage{xcolor}
\usepackage{amsthm}

\makeatletter
\newtheorem*{rep@theorem}{\rep@title}
\newcommand{\newreptheorem}[2]{%
\newenvironment{rep#1}[1]{%
 \def\rep@title{#2 \ref{##1}}%
 \begin{rep@theorem}}%
 {\end{rep@theorem}}}
\makeatother

\newtheorem{lemma}{Lemma}
\newtheorem{theorem}{Theorem}
\newreptheorem{theorem}{Theorem}
\newreptheorem{lemma}{Lemma}

\newtheorem{corollary}{Corollary}

\newtheorem{remark}{Remark}
\newtheorem*{remark*}{Remark}

\usepackage[english]{babel}
\usepackage[utf8]{inputenc}
\usepackage{braket}
\usepackage{amsfonts}%
\usepackage{amssymb}
\usepackage{subcaption}
\usepackage{tabularx}
\usepackage[linesnumbered,ruled,vlined]{algorithm2e}
\usepackage{soul}
\usepackage{bm}

\usepackage{pythonhighlight}

\newtheoremstyle{example}{}{}{}{}{\bfseries}{\smallskip}{\newline}{}
\theoremstyle{example}
\newtheorem{example}{Example}

\usepackage{ifpdf}
\usepackage{epstopdf}

\newcommand{\RM}[0]{\mathrm{RM}}
\newcommand{\T}[0]{\mathrm{T}}

\begin{document}

\title{Entanglement-assisted Quantum Reed-Muller Tensor Product Codes}

\author{Priya~J.~Nadkarni}
\affiliation{Department of Electronic Systems Engineering, Indian Institute of Science, Bengaluru, India, 560012}
\orcid{0000-0002-1351-2959}
\email{priya@alum.iisc.ac.in}
\author{Praveen~Jayakumar}
\email{praveen.jayakumar@mail.utoronto.ca}
\orcid{0000-0002-1523-8260}
\affiliation{Department of Electronic Systems Engineering, Indian Institute of Science, Bengaluru, India, 560012}
\affiliation{Department of Chemistry, University of Toronto, Toronto, Canada, M5S 3H6}
\author{Arpit Behera}
\orcid{0009-0002-1496-4639}
\email{arpitbehera@alum.iisc.ac.in}
\affiliation{Department of Electronic Systems Engineering, Indian Institute of Science, Bengaluru, India, 560012}
\affiliation{Department of Complex Systems, Weizmann Institute of Science, Rehovot, Israel, 7610001}
\author{Shayan Srinivasa Garani}
\email{shayangs@iisc.ac.in}
\orcid{0000-0002-2459-1445}
\affiliation{Department of Electronic Systems Engineering, Indian Institute of Science, Bengaluru, India, 560012}
\thanks{\\{A part of} this work has previously been presented as part of an invited talk at IEEE Information Theory and Applications Workshop 2023.}
\maketitle
\begin{abstract}
We present the construction of {standard} entanglement-assisted (EA) qubit Reed-Muller (RM) codes {and their tensor product variants from classical RM codes}. We show that the EA RM codes obtained using the CSS construction have zero coding rate and negative catalytic rate. We further show that EA codes constructed from these same classical RM codes using the tensor product code (TPC) construction have positive coding rate and provide a subclass of EA RM TPCs that have positive catalytic rate, thus establishing the coding analog of superadditivity for this family of codes, useful towards quantum communications. {We also generalize this analysis to obtain conditions for EA TPCs from classical codes to have positive catalytic rate when their corresponding EA CSS codes have zero rate.}
\end{abstract}
\section{Introduction}

Quantum error correction is essential for building fault-tolerant quantum computing systems and robust quantum communication systems. Shor first constructed a single qubit error correcting quantum code \cite{Shor_QECC}, after which Gottesman proposed the stabilizer framework \cite{Gottesman_thesis} to construct quantum codes. The stabilizer framework based on Weyl-Heisenberg operators from the Pauli group can be considered analogous to the classical additive coding framework. The nature of the stabilizer framework necessitates the stabilizers to commute with each other, enforcing the analogous classical additive code to satisfy a dual-containing constraint. 

Calderbank, Shor, and Steane (CSS) further proposed a method to construct quantum codes, also called CSS codes \cite{CSS_CS}\cite{CSS_Steane}, from two classical codes that satisfy a dual-containing constraint. Since the CSS code properties depend on the corresponding classical codes which are well-studied, the analysis of CSS codes is simple. Brun \textit{et al.} further introduced the construction of quantum codes, also called entanglement-assisted (EA) codes \cite{Brun06}, from classical codes that do not satisfy the dual-containing constraint by introducing the concept of utilizing pre-shared entangled states between the transmitter and receiver. The receiver end qubits of the entangled state are assumed to be noise-free. The EA code construction relies on constructing an abelian group from a set of non-commuting operators. This class of codes can provide better error correction ability than the un-assisted case, useful for EA communications. {EA CSS codes are constructed from two classical codes that do not satisfy the dual-containing criteria \cite{Hsieh07} \cite{Guenda17}.}

Among the various classical codes studied over the years, Reed-Muller (RM) codes have been used in satellite and deep-space communications, {and polar codes, a generalization of RM codes, are being used in the control channel} of the 5G standard \cite{RM_Codes_5G}. Their algebraic properties make them not only locally testable but also locally decodable and list decodable \cite{RM_Codes} \cite{RM_ListDecoding}. {RM codes have soft-decision decoders that utilize the soft information to perform better.
\cite{RM_Code_Soft_Decision}
} {Classical and quantum RM codes are known to achieve the capacity of classical and quantum erasure channels \cite{RM_Codes_CapacityAchieveErasure} \cite{QRM_Codes_CapacityAchieve_Erasure}, respectively}. The construction of binary quantum RM codes from dual-containing classical RM codes was first proposed by Steane \cite{Steane_RM}. Sarvepalli and Klappenecker \cite{QRM_Sarvepalli} generalized the quantum RM code construction to the non-binary case. 

In this paper, we study the construction of EA RM codes from classical RM codes that do not satisfy the dual-containing criteria. We obtain an analytical expression for the number of pre-shared entangled qubits required to construct the code. We further show that the EA RM codes constructed from a single classical RM code have zero coding rate and negative catalytic rate; hence, they might not be useful in practice. 

On the other hand, {classical tensor product codes (TPCs) constructed from two classical codes whose parity check matrix is the tensor product of the parity check matrices of these classical component codes are known to have higher coding rate than their component codes \cite{Wolf06}. } {Using the tensor product construction on the classical codes followed by the CSS/EA CSS construction on the TPC, we obtain the quantum tensor product code (QTPC)/ EA TPC \cite{CodingAnalogSuperadditivity}\cite{QuantumTPC_Fan}\cite{QuantumTPC_Fan2}.  The quantum TPC constructions have} been shown to have the ability to construct positive rate quantum TPCs from classical codes which produce zero rate CSS codes. This property is known as the coding analog of superadditivity\cite{CodingAnalogSuperadditivity}. {In this paper, we show that when the number of pre-shared entangled qubits for each of the EA CSS component codes is lesser than the number of independent parity checks of the classical component codes, coding analog of superadditivity can be achieved. We also show that when $R_1+R_2 > R_1R_2 + 0.5$, the EA TPC obtained has a positive catalytic rate, where $R_1$ and $R_2$ are the coding rates of the classical component codes.}

Using the {EA} TPC construction, we construct the entanglement-assisted RM TPCs from the classical RM codes. We show that all these EA RM {TPCs} have positive coding rate and some of them have positive catalytic rate as well. We further show that EA RM TPCs constructed using both classical components codes to be $\mathrm{RM}(r,m)$ have positive catalytic rates if $m - 2r \leq l(r)$ and have negative catalytic rates if $m - 2r > l(r)$, where $l(r)$ is the maximum value of $i \in \mathbb{Z}^+$ satisfying $\displaystyle\sum_{u=0}^{r}\binom{2r+i}{u} > \frac{2^{2r+i}}{2 + \sqrt{2}}$. {We view $l(r)$ as a natural boundary on $(m-2r)$ that separates EA RM TPCs with positive catalytic rate from those with negative catalytic rate. In Section \ref{sec:tpcexample}, we rigorously derive the expression for $l(r)$.} {We also derive a condition for EA polar TPCs that have positive catalytic rate.} These constructed EA RM TPCs can be used with an optimal quantum interleaver to enhance its burst error correction ability \cite{QTPC_Interleaver}{, where burst errors correspond to errors on contiguous qubits transmitted across a quantum communication channel or stored in a quantum memory}.

The paper is organized as follows: In Section \ref{sec:prelim}, we briefly discuss preliminary concepts such as RM codes, EA CSS codes, classical TPCs, EA TPCs, various interpretations of quantum coding rate for EA codes, and coding analog of superadditivity. {We also derive a condition for an EA TPC to have a positive catalytic rate based on the coding rates of its classical component codes.} In Section \ref{sec:EARM}, we provide the code construction of EA RM codes, provide the explicit form of the number of pre-shared entangled qubits required to construct the code, and prove it has zero coding rate and negative catalytic rate. In Section \ref{sec:EATPC_RM}, we use the EA TPC construction with classical RM codes and show that EA {RM TPCs} with positive coding rate and positive catalytic rate can be constructed. We further provide a subclass of EA RM TPCs whose both coding rate and catalytic rate are positive. We also discuss the relations between EA RM TPCs and larger length EA RM codes. {We end this Section by deriving a condition for EA polar codes to have positive catalytic rates.} Finally, we conclude our paper in Section~\ref{sec:Conclusion}.\vspace{-0.1cm}

\section{Preliminaries} \label{sec:prelim}

\subsection{Reed-Muller codes} \label{subsec:rmcodes}
The classical binary Reed-Muller (RM) code $\RM(r, m)$, is a linear code of length $2^m$, dimensions $\sum_{i=0}^r{m \choose i}$, and distance $2^{m-r}$, i.e., it is a $[2^m, \sum_{i=0}^{r}{m \choose i}, 2^{m-r}]$ linear code. The codewords of the code can be obtained as evaluations of polynomials. We define the \textit{evaluation vector} of a multivariate polynomial $f: \mathbb{F}_2^m \rightarrow \mathbb{F}_2$ from the polynomial ring $W_m := \mathbb{F}_2[x_1, x_2, \cdots, x_m]$ as the $2^m$-length vector obtained by the concatenation of the bit values of {$f$} evaluated at all points $\bm{z} = (z_1, z_2, \cdots, z_m)$ in the domain $\mathbb{F}_2^m$.
If we denote the evaluation of a single point as
\begin{equation}
    \mathrm{Eval}_{\bm{z}}(f) := f(\bm{z}),
\end{equation}
then, 
\begin{equation}\label{def:eval}
    \mathrm{Eval}^{(m)}(f) := (\mathrm{Eval}_{\bm{z}}(f) | \forall \bm{z} \in \mathbb{F}_2^m).
\end{equation}
Generally, in equation \eqref{def:eval} we consider the points of evaluation $\bm z$ in an increasing order, where the ordering is in the usual sense of binary numbers. This choice becomes important to obtain the well-known Plotkin construction of the code \cite{Abbe}.

The RM code $\RM(r, m)$ is the set of evaluation vectors of all multivariate polynomials over $m$ variables upto degree $r$.
\begin{equation}\label{def:RMEval}
    \mathrm{RM}(r, m) := \{\mathrm{Eval}^{(m)}(f)|\forall f\in W_m, \mathrm{deg}(f)\leq r\}.
\end{equation}

\begin{example} \label{ex:RM_code_example}
 Let us consider $r=1$ and $m=4$. The RM code is a $[2^4, \underset{i=0}{\overset{1}{\sum}} {4 \choose i}, 2^{4-1}] = [16, 5, 8]$ code, {given by the set of evaluation vectors of all multivariate polynomials over $4$ variables upto degree $1$.} We note that the set of all multivariate monomials over $m$ variables upto degree $r$ forms the basis of the set of all multivariate polynomials over $m$ variables upto degree $r$. Thus, the set of all monomials in $W_4$ upto degree $1$ forms the basis of the set of all polynomials in $W_4$ upto degree $1$. For the multivariate polynomial {of degree at most $1$} defined over the variables $x_1$, $x_2$, $x_3$, and $x_4$, the basis of multivariate monomials is $\{1, x_1, x_2, x_3, x_4\}$. The evaluation vector $\mathrm{Eval}^{(4)}(x_1)$ of the monomial $x_1$ is $[0~0~0~0~0~0~0~0~1~1~1~1~1~1~1~1]$ as the value of $x_1$ is $0$ for the first $8$ vectors in $\mathbb{F}_2^{m}$ and is $1$ for the rest $8$ vectors. Similarly, we obtain the evaluations $\mathrm{Eval}^{(4)}(1) = [1~1~1~1~1~1~1~1~1~1~1~1~1~1~1~1]$, $\mathrm{Eval}^{(4)}(x_2) = [0~0~0~0~1~1~1~1~0~0~0~0~1~1~1~1]$, $\mathrm{Eval}^{(4)}(x_3) = [0~0~1~1~0~0~1~1~0~0~1~1~0~0~1~1]$, and 
 $\mathrm{Eval}^{(4)}(x_4) = [0~1~0~1~0~1~0~1~0~1~0~1~0~1~0~1]$. The generator matrix of the $\mathrm{RM}(1,4)$ code is
 \begin{align}
\!\!\!G\! \! &=\!\!\!\begin{bmatrix}\mathrm{Eval}^{(4)}(1)\\\mathrm{Eval}^{(4)}(x_1)\\\mathrm{Eval}^{(4)}(x_2)\\\mathrm{Eval}^{(4)}(x_3)\\\mathrm{Eval}^{(4)}(x_4)\end{bmatrix},\nonumber\\
&=\!\!\!\left[\!\begin{array}{p{0.05cm}p{0.05cm}p{0.05cm}p{0.05cm}p{0.05cm}p{0.05cm}p{0.05cm}p{0.05cm}p{0.05cm}p{0.05cm}p{0.05cm}p{0.05cm}p{0.05cm}p{0.05cm}p{0.05cm}p{0.05cm}}
1 & 1 & 1 & 1 & 1 & 1 & 1 & 1 & 1 & 1 & 1 & 1 & 1 & 1 & 1 & 1\\
0 & 0 & 0 & 0 & 0 & 0 & 0 & 0 & 1 & 1 & 1 & 1 & 1 & 1 & 1 & 1\\
0 & 0 & 0 & 0 & 1 & 1 & 1 & 1 & 0 & 0 & 0 & 0 & 1 & 1 & 1 & 1\\
0 & 0 & 1 & 1 & 0 & 0 & 1 & 1 & 0 & 0 & 1 & 1 & 0 & 0 & 1 & 1\\
0 & 1 & 0 & 1 & 0 & 1 & 0 & 1 & 0 & 1 & 0 & 1 & 0 & 1 & 0 & 1
\end{array}\right]\!\!.\label{eqn:RM_Example_G}
\end{align}
We note that the rows of the generator matrix $G$ of the $\mathrm{RM}(1,4)$ code is a subset of the rows of the matrix $\begin{bmatrix}1 & 1\\0 & 1\end{bmatrix}^{\otimes 4}$. 
\end{example}

We next provide a few properties of the evaluation function.
We note that $\mathrm{Eval}^{(m)}(f)$ is linear in the polynomial $f$.
\begin{lemma}[Linearity]\label{lemma:linearity} 
    For $f, g : \mathbb{F}_2^m \rightarrow \mathbb{F}_2; f, g \in W_m := \mathbb{F}_2[x_1, x_2, \cdots, x_m]$, the evaluation function has the following properties:
\begin{subequations}
    \begin{align}
    \!\!(i)& \mathrm{Eval}^{(m)}(f \!+\! g) \!=\! \mathrm{Eval}^{(m)}(f) \!\oplus\! \mathrm{Eval}^{(m)}(g)\!\label{prop:eval1}\\
    \!\!(ii)& \mathrm{Eval}^{(m)}(fg) \!=\! \mathrm{Eval}^{(m)}(f)\!\odot\!\mathrm{Eval}^{(m)}(g)\!\label{prop:eval2}
\end{align}
\end{subequations}
where {$\oplus$ denotes the bitwise addition of two vectors, $\odot$ denotes the bitwise product of the two vectors, and $\oplus, \odot : \mathbb{F}_2^{2^m} \times \mathbb{F}_2^{2^m} \rightarrow \mathbb{F}_2^{2^m}$.}
\end{lemma}
The proof of Lemma \ref{lemma:linearity} is provided in Appendix~\ref{app:prop}.

Since any multivariate polynomial can be written as a linear sum of multivariate monomials, we note that the minimal generators of the codespace are given by the evaluation vectors of monomials in $W_m$. 
We define the set $[m] := \{1, 2, \cdots, m\}$ as the set of positive integers upto $m$. We introduce a shorthand notation for a monomial 
\begin{equation}
    x_A := \prod_{i\in A}x_i, ~~~ A {\subseteq} [m]
\end{equation}
as the product of the variables given by indices in set A. 
The generator $G(r, m)$ of $\RM(r, m)$ is given by
\begin{equation}
    \!\!\!G(r, m) := \{\mathrm{Eval}^{(m)}(x_A) | \forall A {\subseteq} [m], |A| \leq r\}.\!
\end{equation}
By construction, $\RM(0, m)\subset \RM(1, m)\subset \cdots \subset \RM(m, m)$ and $\RM(r, m)^{\perp} = \RM(m-r-1, m)$ {for $m\in \mathbb{Z}^+$.}

The RM code can also be constructed by the Kronecker product construction\cite{Abbe}, wherein the rows of the generator matrix of $\RM(r, m)$ are given by the rows of weight $\geq 2^{m-r}$ of the matrix $\begin{pmatrix} 1& 1 \\ 0 & 1\end{pmatrix}^{\otimes n}$. This suggests that RM codewords are tensor products of smaller RM codewords. We now formalize this notion using the evaluation vector function.

For the bit string in equation \eqref{def:eval}, due to our choice of ordering of $\bm z$, we observe the tensor product property of $\mathrm{Eval}^{(m)}(\cdot)$ 
\begin{lemma}[Tensor product]\label{lemma:tpRM}
    The tensor product of evaluation vectors of polynomials over the field $\mathbb{F}_2$ is an evaluation vector of a polynomial from a larger ring of the same field.
    In particular, for positive integers $m_1, m_2$ and $i \in  [m_1], j \in [m_2]$ where $[m] = \{1, 2, \cdots, m\}$ the evaluation vector has the property:
    \begin{align}
        (i) & \mathrm{Eval}^{(m_1)}(x_i) \otimes \mathrm{Eval}^{(m_2)}(x_j) \nonumber\\ &\hspace{2.4cm}= \mathrm{Eval}^{(m_1 + m_2)}(x_ix_{m_1 + j})\label{tp1}\\
        (ii) & \mathrm{Eval}^{(m_1)}(x_A)\! \otimes\! \mathrm{Eval}^{(m_2)}(x_B) \!\nonumber\\ &\hspace{2.3cm}=\! \mathrm{Eval}^{(m_1 + m_2)}(x_Ax_{m_1 + B})\label{tp2}
    \end{align}
    where $A \subseteq [m_1], B \subseteq [m_2]$ are index sets and $m_1 + B = \{m_1 + b | b \in B\}$ is a shift of indices in set $B$.
\end{lemma}
\hspace{-0.4cm}The proof of Lemma \ref{lemma:tpRM} is provided in Appendix~\ref{app:prop}

Due to the above property and linearity of $\mathrm{Eval}^{(m)}(\cdot)$, tensor product of any codewords of RM codes (not necessarily of the same code parameters $r, m$) are codewords of some larger RM code. Thus, the classical product code generated using two RM codes will also be an RM code. In particular, tensor product of codewords of $\RM(r_1, m_1)$ and $\RM(r_2, m_2)$ will be a codeword in $\RM(r_1 + r_2, m_1 + m_2)$. In Appendix \ref{app:product}, we show that the generator matrix of the product code can be written as the tensor product of the individual generator matrices of the component codes. Thus the product code\cite{MacWilliams} generated using two RM codes is a subset of some larger RM code.
\begin{equation}\label{eqn:tprm_in_rm}
    \RM(r_1, m_1)\otimes\RM(r_2, m_2) {\subseteq} \RM(r_1 + r_2, m_1 + m_2)
\end{equation}
with equivalence when either $(r_1, m_1)$ or $(r_2, m_2)$ is $(0, 0)$.

Conversely, while an RM codeword in $\RM(r, m)$ may not always be expressible as a tensor product of RM codewords of $\RM(r_1, m_1)$ and $\RM(r_2, m_2)$ for $ r_1< r, r_2\leq r ~\mathrm{or}~ r_1\leq r, r_2<r$ and $m_1 + m_2 = m$, the codeword can be written as a linear sum of tensor product of such codewords.
\begin{equation}
    \!\mathrm{Eval}^{(m_1 \!+\! m_2)}(f)\!\! =\!\! \bigoplus_{i}\! \mathrm{Eval}^{(m_1)}(x_{A_i})\otimes \mathrm{Eval}^{(m_2)}(x_{B_i}) 
\end{equation}
where $f = \sum_i x_{A_i}x_{B_i}$ for some index sets $A_i \subseteq [m_1], B_i \subseteq [m_2]$ determined from $f$. In particular, since codewords of $\RM(r, m)$ are given by evaluation vectors of polynomials $f$ of degree $\leq r$, they can be written as a linear sum of codewords from $\RM(\min(r, m_1), m_1)\otimes \RM(\min(r, m_2), m_2)$ as the defining sets $A_i$ ($B_i$) can only have $\leq \mathrm{min}(r, m_1)$ ($\leq \mathrm{min}(r, m_2)$) index elements. For $r < m_1 + m_2$, we have
\begin{align}
    &\RM(r, m_1 + m_2) \nonumber\\
    &~~~{\subseteq} \RM(\min(r, m_1), m_1)\otimes \RM(\min(r, m_2), m_2)\label{eqn:rm_in_tprm}
\end{align}
{with equivalence when $r = 0$.}
Later in Section~\ref{subsec:TPC_QRM_compare}, we use equations \eqref{eqn:tprm_in_rm} and \eqref{eqn:rm_in_tprm} to discuss the relationship between the RM codes and the RM TPCs.


\subsection{CSS codes and entanglement-assistance} \label{subsec:csscodes}
A quantum code {can be} constructed using two classical codes $C_1, C_2$ using the CSS code construction\cite{CSS_CS}. We denote this as $\mathrm{CSS}(C_1, C_2)$. This construction requires that the two codes satisfy the dual-containing criteria, i.e., {$C_1^{\perp} \subseteq \mathrm{C_2}$} and $H_1H_2^{\T} = 0$, where $H_1$ and $H_2$ are parity check matrices of $C_1$ and $C_2$. An encoded quantum state is given by
\begin{equation}\label{def:csscode}
    {\ket{\bm{x}\oplus C_1^\perp} = \frac{1}{\sqrt{|C_1^\perp|}}\sum\limits_{\bm{c} \in C_1^\perp} \ket{\bm{x} \oplus \bm{c}}, ~\bm{x} \in C_2},
\end{equation}
where {$|C_1^\perp| = 2^{ n -k_1}$}.
The codewords of the CSS code are given by a superposition of all codewords in a coset of the {codespace $C_1^\perp$ in $C_2$}.

The parity check matrix of the resulting CSS code has the form: 
\begin{equation}\label{def:CSScheckmatrix}
    H_{CSS}=\left[\begin{array}{c|c} H_1 & \mathbf{0} \\ \mathbf{0} & H_2\end{array}\right],
\end{equation}
where $H_1, H_2$ are the parity check matrices of the classical codes $C_1, C_2$. The vertical line is used to {split each row into equal halves, dividing the columns into two sets of the same size}. The rows of $H_{CSS}$ describe stabilizer operators $\hat S$ such that $\hat S\ket{\psi}= \ket{\psi}$ for $\ket{\psi} \in \mathrm{CSS}(C_1, C_2)$. $\mathrm{CSS}(C_1, C_2)$ denotes the subspace of the Hilbert space invariant by the action of these operators. This space is said to be \textit{stabilized} by the stabilizer operators $\hat S$. For a CSS code over $n$ qubits, $C_1, C_2$ are codes in $\mathbb{F}_2^n$ and a row of the parity check matrix $\bm r = (\bm x, \bm z)$ of length $2n$ corresponds to the stabilizer {$\hat S(\bm r) = \mathrm{i}^{\bm x\cdot \bm z}\hat X^{x_0}\hat Z^{z_0}\otimes \hat X^{x_1}\hat Z^{z_1}\otimes \cdots \otimes \hat X^{x_{n-1}}\hat Z^{z_{n-1}}$}, i.e., a binary to Pauli isomorphism from a binary tuple of length $2n$ to a {Hermitian} Pauli operator comprising $X$ and $Z$ components over each qubit, {where $\bm x\cdot \bm z$ denotes the scalar product of $\bm x$ and $\bm z$ and $\mathrm{i} = \sqrt{-1}$}.

However, when the dual-containing criteria is not satisfied {($C_1^{\perp}{\not\subseteq} C_2$)} and the CSS code construction fails, we need to use the EA CSS construction \cite{Wilde08}. An EA CSS code constructed from an $[n,k,d]$ code is an $[[n, 2k-n + n_e, \geq d; n_e]]$ quantum code (or $[[n, n - 2\rho+ n_e, \geq d; n_e]]$ quantum code), which requires $n_e = \mathrm{gfrank}(HH^{\mathrm{T}})$ entangled qubit pairs, where $\rho = n-k$ and $H$ is the parity check matrix of the classical code, and $\mathrm{gfrank(\cdot)}$ is the rank of the matrix over the Galois field $\mathbb{F}_2$. 

At the transmitter end, $n - 2\rho + n_e$ information qubits are encoded onto $n$ qubits, and these $n$ qubits are transmitted in a single channel use. $n_e$ qubits of the $n-2\rho + n_e$ logical qubits\footnote{{The number of logical qubits is the number of message qubits transmitted when a codeword is transmitted through the channel.}} are from the shared entangled pairs. These can be generated \textit{a~ priori} and shared across the transceiver, without adding runtime latencies. The entangled qubits are used to systematically restore the dual-containing criteria in a higher-dimensional codespace whose stabilizers are obtained by optimally extending the operators obtained from the classical parity check matrix using symplectic Gram-Schmidt orthogonalization \cite{Brun06} \cite{Wilde08}\cite{Nadkarni_NBEASC}.

\subsection{Classical and Quantum Tensor Product Codes} \label{subsec:cqtpc}
A classical TPC \cite{Classical_TPC} obtained from two classical codes $C_1[n_1,k_1,d_1]$ and $C_2[n_2,k_2,d_2]$ whose parity check matrices are $H_1$ and $H_2$ is the code whose parity check matrix is $H = H_1 \otimes H_2$. The size of $H_i$ is $\rho_i \times n_i$, for $i=1,2$, where $\rho_i = n_i-k_i$. The distance of the TPC is known to be the minimum of $d_1$ and $d_2$. The parameters of the TPC $C$ is
\begin{align*}
{[n_1n_2, n_1n_2 - \rho_1\rho_2, \mathrm{min}(d_1,d_2)]}.
\end{align*}

Quantum TPCs can be constructed from classical TPCs using the CSS construction when the classical TPC is a dual-containing code and using EA CSS construction when the classical TPC is not a dual-containing code\cite{CodingAnalogSuperadditivity} \cite{EABQTPC}\cite{QuantumTPC_Fan}. We note that when the classical TPC and its component codes are defined over the same field, the classical TPC is a dual-containing code if and only if one of its component codes is a dual-containing code as $HH^{\T} = H_1H_1^{\T} \otimes H_2H_2^{\T}$\cite{EABQTPC}.

The code parameters of the quantum TPC $\mathcal{C}_{\mathrm{TPC}}$ constructed from $C_1$ and $C_2$ defined over $\mathbb{F}_2$ is \begin{align*}
    {[[n_1n_2, n_1n_2 - 2\rho_1\rho_2 {+ n_e}, \geq \mathrm{min}(d_1,d_2); n_e]]},
\end{align*}
where $n_1n_2$ is the {transmitter end qubits of the codeword}, $n_1n_2 - 2\rho_1\rho_2$ is the number of logical qubits, the code distance is at least $\mathrm{min}(d_1,d_2)$, and the number of pre-shared entangled qubits $n_e =  \mathrm{gfrank}(H_1H_1^{\T})\mathrm{gfrank}(H_2H_2^{\T})$. {The length of the code is $n_1n_2+n_e$. We note that $n_e=0$ when the classical TPC is a dual-containing code.} We note that $n_e = n_{e1}n_{e2}$, where $n_{ei} = \mathrm{gfrank}(H_iH_i^{\T})$ is the minimal number of pre-shared entangled qubit required to construct the EA CSS code from the classical component code $C_i$, for $i=1,2$.

\subsection{Various Interpretations of Coding Rates of Entanglement-Assisted CSS Codes}
In an EA CSS code {defined by $s$ independent stabilizer generators}, one qubit of the pre-shared entangled state is assumed to be at the receiver end maintained error-free and the other qubit is a part of the $n$ qubits of the $(n+n_e)$ qubit codeword at the transmitter end. This leads to the following interpretations of the coding rates:
\begin{itemize}
\item[1)] Entanglement-assisted rate: This is the ratio of the number of logical qubits to the number of qubits transmitted across the quantum channel, i.e., $\mathcal{R}_E = \frac{n+n_e - {s}}{n}$. It is assumed here that the entanglement between the sender and receiver is free \cite{Wilde08}, where the entangled qubits were shared between the transmitter and the receiver \textit{a~priori} from a Bell pair source as shown in Figure \ref{fig:EA}(a).
\item[2)] Trade-off rate: This is the pair of ratios, {given by the ratio of the number of logical qubits to the number of qubits transmitted across the quantum channel and the number of pre-shared entangled pairs to the number of qubits transmitted across the quantum channel,} i.e., $\mathcal{R}_T = (\frac{n+n_e - {s}}{n}, \frac{n_e}{n})$. The first term of the pair is the rate at which noiseless qubits are generated, and the second term of the pair is the rate at which entangled qubits are consumed \cite{Wilde08}. 
\item[3)] Catalytic rate: This is the ratio of the difference in the number of logical qubits and the number of pre-shared entangled pairs to the number of qubits transmitted across the quantum channel, i.e., $\mathcal{R}_C = \frac{n + n_e - {s} - n_e}{n} = \frac{n - {s}}{n}$. It is assumed here that the receiver end qubits of the pre-shared entangled pairs {are} built up at the expense of the transmitted logical qubits. {We illustrate this in Figure~\ref{fig:EA}(b). The receiver end qubits at time instant $t$ are given by $\ket{\eta}_{R,t}^{\otimes n_e}$. They belong to the Bell pairs $\ket{\eta}_{t}^{\otimes n_e}$, and are considered to be transmitted along with the message qubits of the previous codeword at time $(t-1)$. Through the encoding and error correction of the codeword at time $(t-1)$, the receiver end qubits $\ket{\eta}_{R,t}^{\otimes n_e}$ are maintained error-free and used for decoding by the receiver.}
\end{itemize}

\begin{figure*}[t]
    \centering
    \includegraphics[width = \textwidth]{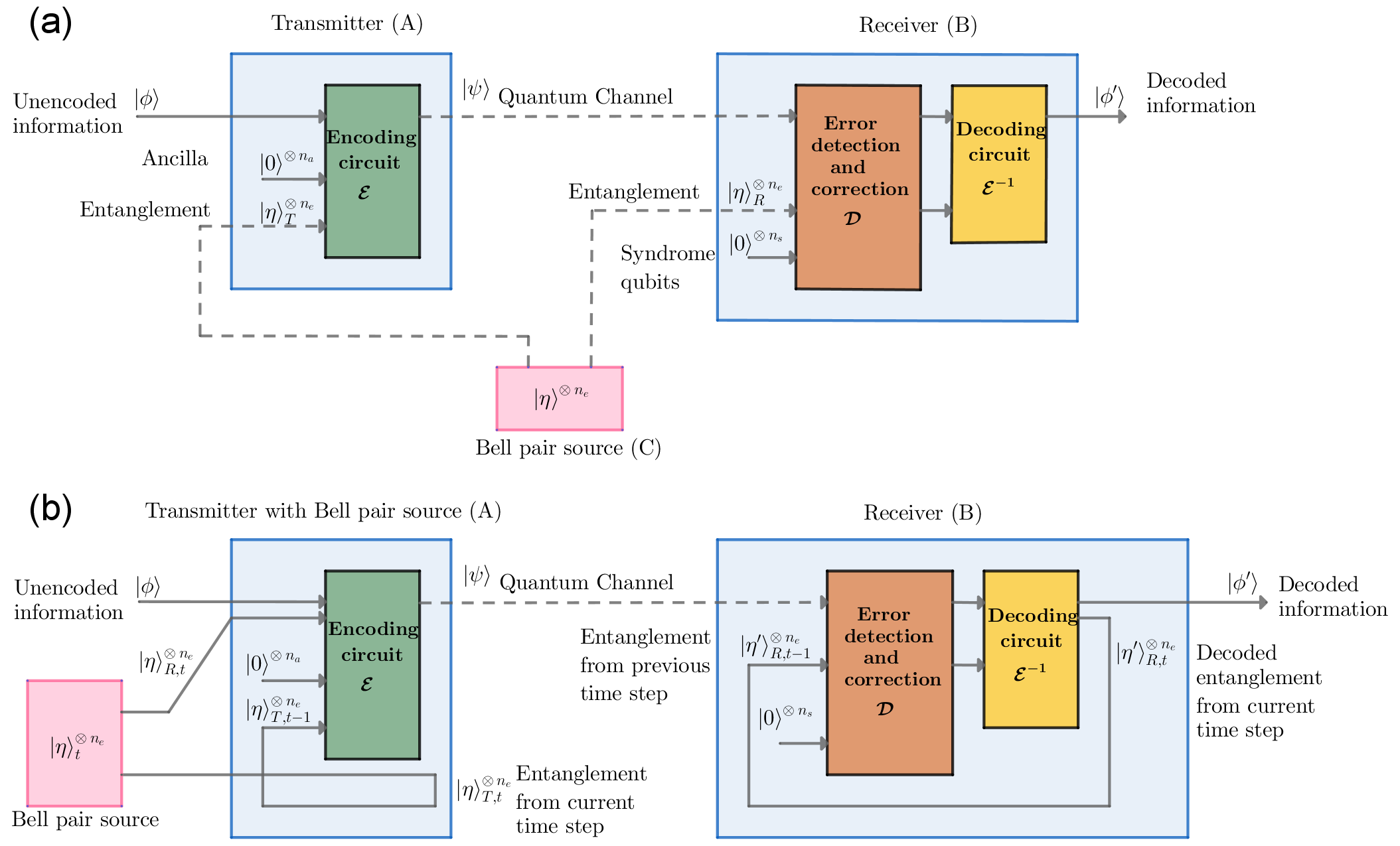}
    \caption{
    Entanglement-assisted quantum error correction schemes between the transmitter Alice (A) and the receiver Bob (B), respectively, when (a) the Bell pairs $\ket{\eta}^{\otimes n_e}$ are freely available between the transmitter $\ket{\eta}_T^{\otimes n_e}$ and the receiver $\ket{\eta}_R^{\otimes n_e}$ which was generated and distributed by a third party Charlie (C) beforehand, where the subscript T (R) denotes transmitter (receiver) qubits of the Bell pairs{;} and (b) entanglement between the transmitter and receiver is catalytically built up by generating $n_e$ Bell pairs $\ket{\eta}_{t}^{\otimes n_e}$ in each time step $t$. {At time step $(t-1)$, this scheme consumes} $n_e$ qubits of the unencoded information to store {the receiver end} qubits {of the} $n_e$ Bell pairs, denoted by $\ket{\eta}_{R, t}^{\otimes n_e}$, in a reliable fashion. {The codeword is obtained at the transmitter by encoding $\ket{\eta}_{R, t}^{\otimes n_e}$ with the rest of the $n_e$ qubits $\ket{\eta}_{T, t-1}^{\otimes n_e}$ of the shared Bell pairs from time step $(t-1)$ which are used to obtain the codeword along with $n_a$ ancilla qubits in state $\ket{0}$ and the unencoded $(n-2\rho)$-qubit message information $\ket{\phi}$.} The error correction and decoding at time step $t$ uses the receiver end qubits $\ket{\eta'}_{R, t-1}^{\otimes n_e}$ of the Bell pairs obtained from the decoded information from the codeword at time step $(t-1)$. We note that $\ket{\eta'}_{R, t-1}^{\otimes n_e}$ is used instead of $\ket{\eta}_{R, t-1}^{\otimes n_e}$ as it is an estimate of the unencoded information. For both cases, $n_s$ syndrome qubits prepared in state $\ket{0}$ are required at the receiver end to detect and correct errors in transmission.}
    \label{fig:EA}
\end{figure*}
\subsection{Coding Analog of Superadditivity}
  {Let $C_1[n_1,k_1,d_1]$ and $C_2[n_2,k_2,d_2]$ be the classical component codes from which the EA CSS codes and EA TPCs are constructed. Let $\rho_i = n_i-k_i$, for $i = 1,2$.}
  The coding rate of the classical {TPC based on $C_1$ and $C_2$}, i.e, $1 - \frac{\rho_1\rho_2}{n_1n_2}$, is greater than the coding rates of the classical component codes {$C_1$ and $C_2$}, i.e., $1-\frac{\rho_1}{n_1}$ and $1-\frac{\rho_2}{n_2}$, as $\rho_1 < n_1$ and $\rho_2 < n_2$. We next provide a similar result for the coding rate of quantum{/EA} TPCs compared to the coding rate of the CSS{/}EA CSS codes obtained from the classical TPC's component codes.
  \begin{lemma}\label{lem:EATPC_CodingAnalog}
      The coding rate/entanglement-assisted rate of the quantum{/EA} TPC, i.e, $1+\frac{n_e-2\rho_1\rho_2}{n_1n_2}$, is at least the maximum of the coding rates of the CSS/EA CSS codes obtained from the classical component codes, i.e., $1 {+} \frac{n_{e1} - 2\rho_1}{n_1}$ and $1 {+} \frac{n_{e2} - 2\rho_2}{n_2}$, respectively. When $n_{e1} = \mathrm{gfrank}(H_1H_1)^{\mathrm{T}} < \rho_1$ and $n_{e2} = \mathrm{gfrank}(H_2H_2)^{\mathrm{T}} < \rho_2$, the coding rate/entanglement-assisted rate of the quantum{/EA} TPC is greater than the coding rates of the {CSS/}EA CSS codes based on the classical TPC's component codes.
  \end{lemma}
  \begin{proof}
   We note that the number of logical qubits of the EA CSS code is $(n_i + n_{ei} - 2\rho_i) \geq 0$ for $i=1,2$. Thus, we obtain
  \begin{align}
  & n_2 + n_{e2} - 2\rho_2 \geq 0 \Rightarrow \rho_1(n_2 + n_{e2} - 2\rho_2) \geq 0,\nonumber\\
  \Rightarrow&2\rho_1(n_2 - \rho_2) -\rho_1(n_2 - n_{e2}) \geq 0,\nonumber\\
  \Rightarrow&2\rho_1(n_2 - \rho_2) -n_{e1}(n_2 - n_{e2}) \geq 0, \nonumber\\
  &~~~~~~~~~~~~~~~~~~~~~(\because n_{e1} =\mathrm{gfrank}(H_1H_1^{\mathrm{T}}) \leq \rho_1 )\nonumber\\
  \Rightarrow& n_2(2\rho_1 - n_{e1}) \geq  2\rho_1\rho_2 -n_{e1}n_{e2}, \nonumber\\
  \Rightarrow& \frac{(2\rho_1 - n_{e1})}{n_1} \geq  \frac{2\rho_1\rho_2 -n_{e1}n_{e2}}{n_1n_2}, \nonumber\\
  \Rightarrow& 1 - \frac{(2\rho_1 - n_{e1})}{n_1} \leq 1- \frac{2\rho_1\rho_2 -n_{e1}n_{e2}}{n_1n_2}, \nonumber
  \end{align}
  \begin{align}
  \Rightarrow& \frac{n_1 \!+\! n_{e1}\! -\! 2\rho_1}{n_1} \leq \frac{n_1n_2\! +\! n_{e1}n_{e2}\! -\! 2\rho_1\rho_2}{n_1n_2}. \label{eqn:CodingAnalog1}
  \end{align}
    We note that the inequality in equation \eqref{eqn:CodingAnalog1} becomes a strict inequality when $n_{e1} =\mathrm{gfrank}(H_1H_1^{\mathrm{T}}) < \rho_1$.
  Similarly, we have 
  \begin{align}
      \frac{n_2 + n_{e2} - 2\rho_2}{n_2} \leq \frac{n_1n_2 + n_{e1}n_{e2} - 2\rho_1\rho_2}{n_1n_2}. \label{eqn:CodingAnalog2}
  \end{align}
  From equations \eqref{eqn:CodingAnalog1} and \eqref{eqn:CodingAnalog2}, we obtain
  \begin{align}
      \!&\frac{n_1n_2 + n_{e1}n_{e2} - 2\rho_1\rho_2}{n_1n_2} \nonumber\\
      \!&\geq \mathrm{max}\!\left(\frac{n_1 + n_{e1} - 2\rho_1}{n_1}, \frac{n_2 + n_{e2} - 2\rho_2}{n_2}\right)\!. \label{eqn:CodingAnalog}
  \end{align}
  When $n_{e1} =\mathrm{gfrank}(H_1H_1^{\mathrm{T}}) < \rho_1$ and $n_{e2} =\mathrm{gfrank}(H_2H_2^{\mathrm{T}}) < \rho_2$, the inequalities in equations \eqref{eqn:CodingAnalog1} and \eqref{eqn:CodingAnalog2} become strict inequalities, and the coding rate/entanglement-assisted rate of the quantum{/EA} TPC is greater than those of the {CSS/}EA CSS codes based on the classical component codes, i.e., $\frac{n_1 + n_{e1} - 2\rho_1}{n_1}< \frac{n_1n_2 + n_{e1}n_{e2} - 2\rho_1\rho_2}{n_1n_2}$ and $\frac{n_2 + n_{e2} - 2\rho_2}{n_2} < \frac{n_1n_2 + n_{e1}n_{e2} - 2\rho_1\rho_2}{n_1n_2}$. 
  \end{proof}
  \begin{corollary}\label{coro:CodingLawPossibility}
  Let the coding rate/entanglement-assisted rate of EA CSS codes constructed from classical codes $C_1$ and $C_2$ be $0$. {If $C_1$ and $C_2$ have positive coding rate}, then the quantum{/EA} TPC constructed considering $C_1$ and $C_2$ as classical component codes have positive rate.
  \end{corollary}
  \begin{proof}
      {From Lemma \ref{lem:EATPC_CodingAnalog}, if $C_1$ and $C_2$ satisfy $\mathrm{gfrank}(H_1H_1^{\mathrm{T}}) < \rho_1$ and $\mathrm{gfrank}(H_2H_2^{\mathrm{T}}) < \rho_2$, the quantum/EA TPC has a positive rate. As the coding rates/entanglement-assisted rates of the EA CSS codes are $0$, for $i=1,2$, $2k_i - n_i + \mathrm{gfrank}(H_1H_1^{\mathrm{T}}) = 0 \Rightarrow \mathrm{gfrank}(H_1H_1^{\mathrm{T}}) = \rho_i - k_i < \rho_i$ when $k_i > 0$. Thus, when $C_1$ and $C_2$ have positive coding rate, i.e., $k_1/n_1, k_2/n_2 > 0$, the quantum{/EA} TPC constructed has a positive rate.}
  \end{proof}
  
  From Corollary \ref{coro:CodingLawPossibility}, we note that we can construct positive rate EA TPCs from classical TPCs whose component codes {have positive coding rate} and yield zero-rate EA CSS codes. This phenomenon is known as the coding law of superadditivity, that has been first reported in \cite{CodingAnalogSuperadditivity}. We show in Section \ref{sec:EARM} that the EA CSS codes obtained from classical RM codes satisfy $n_{e1} < \rho_1$, $n_{e2} < \rho_2$, and $k_1, k_2 >0$; hence, the EA TPC obtained from these RM codes have positive coding rate as proved rigorously in Section \ref{sec:EATPC_RM}.

  {We note that positive entanglement-assisted rate of a code does not ensure positive catalytic rate. In Corollary \ref{coro:EACSS_positive_catalytic}, we provide a condition for the classical component codes of the EA TPC to satisfy that ensures its positive catalytic rate.}

{\begin{corollary}[Positive catalytic rate of EA TPC] \label{coro:EACSS_positive_catalytic}
    Let the entanglement-assisted rate of EA CSS codes constructed from $C_1$ and $C_2$ be $0$. Let $R_1$ and $R_2$ be the coding rates of the classical codes $C_1$ and $C_2$. Then, the EA TPC constructed from the TPC based on $C_1$ and $C_2$ has a positive catalytic rate when $R_1 + R_2 > R_1R_2 + 0.5$.
\end{corollary}} 

\begin{proof}
    We note that when the EA CSS codes have zero coding rates, $n_i + n_{ei} = 2(n_i - k_i) \Rightarrow n_{ei} = n_i - 2k_i$, $i = 1,2$. The catalytic rate of the EA TPC code is $\frac{(n_1n_2 - 2(n_1-k_1)(n_2-k_2))}{n_1n_2} = \frac{2n_1k_2 + 2n_2k_1 - 2k_1k_2 - n_1n_2}{n_1n_2} = 2R_2 + 2R_1 - 2R_1R_2 - 1$.  For $i=1,2$, let $R_i = k_i/n_i$ be the coding rate of the classical component code $C_i$. The EA TPC has a positive catalytic rate when
    \begin{align}
        &2R_1 + 2R_2 - 2R_1R_2 -1 > 0,\nonumber\\
    \Rightarrow & R_1 + R_2 > R_1R_2 + 0.5.\nonumber \qedhere
    \end{align}
\end{proof}

\begin{figure}[t]
    \centering
    \includegraphics[width = \columnwidth]{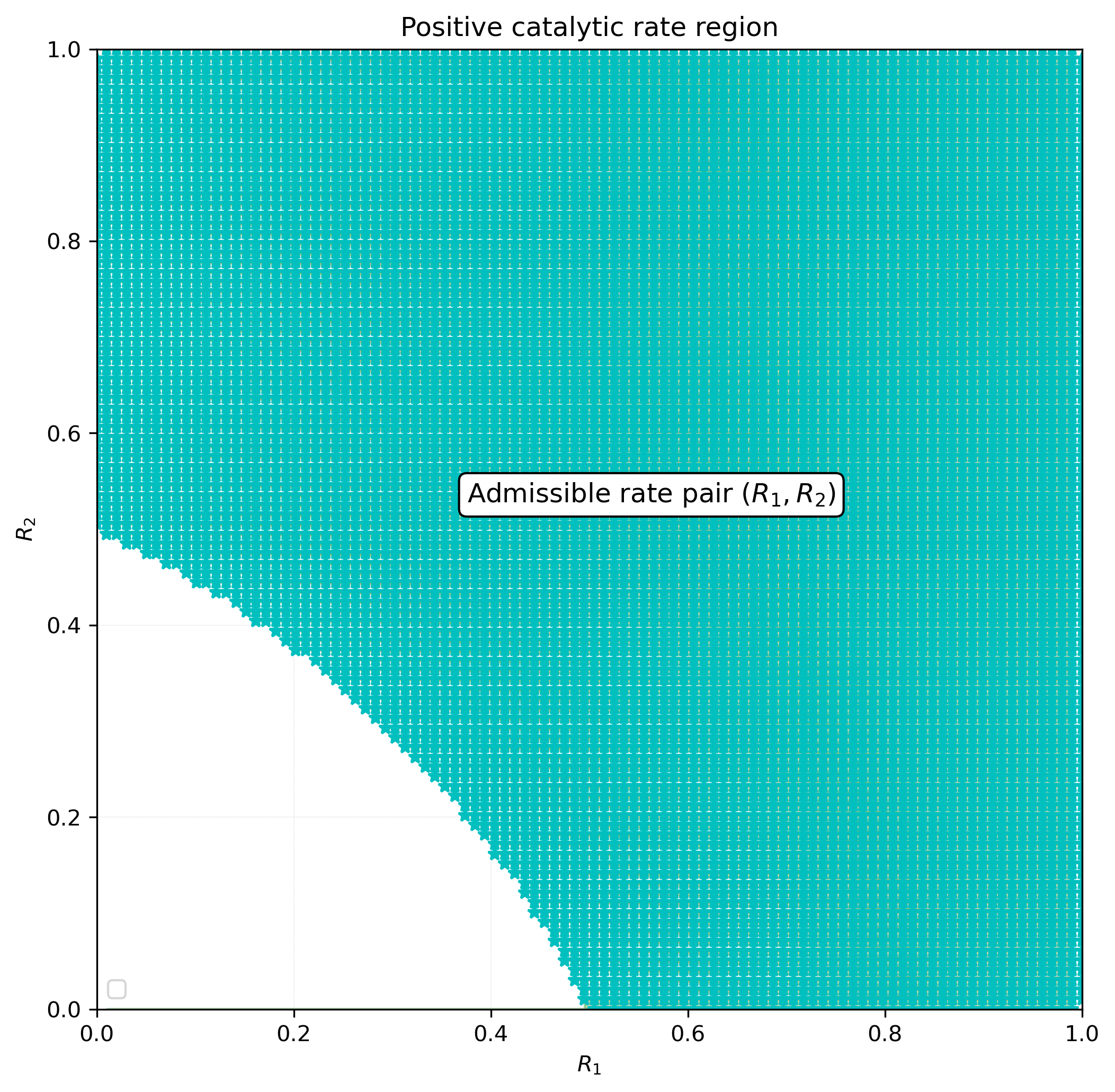}
    \caption{Admissible rate pairs for positive catalytic rate of EA TPC codes. When both the coding rates are less than 0.5, there is a admissible region for positive catalytic rates as shown in the shaded region. When one of the component coding rates is $>0.5$, the inequality is always valid.}
    \label{fig:pcr}
\end{figure}

{Figure \ref{fig:pcr} shows the admissible rate pairs to achieve positive catalytic rates for EA TPCs based on the conditions in Corollary \ref{coro:EACSS_positive_catalytic}. When one of the component coding rates is $>0.5$, the inequality is always true, implying a wide range of rate pairs that yield positive catalytic rate, useful to practice.}

\section{Entanglement-assisted Reed-Muller Code} \label{sec:EARM}
For RM codes, since $\RM(r, m)^{\perp} = \RM(m-r-1, m)$, we can construct CSS codes using $C_1 = \RM(r_1, m)$ and $C_2 = \RM(r_2, m)$ for $m - r_2 - 1 \leq r_1$. When $r_1 = r_2 = r$, this condition becomes $r \geq \frac{m-1}{2}$. We refer to this CSS code as the Quantum Reed-Muller code $\mathrm{QRM}(r, m)$.
For $r < \frac{m-1}{2}$, we require to use the entanglement-assisted CSS construction to obtain the {EA RM code}. {Throughout this paper, we call the quantum RM codes that do not use entanglement-assistance as QRM codes.} 

\subsection{Number of pre-shared entangled qubits, $n_e$}
The number of pre-shared entangled qubits $n_e$ required to construct the entanglement-assisted CSS code using a classical code that has the check matrix $H$ is given by the expression $n_e = \mathrm{gfrank}(HH^\T)$.

\begin{lemma}\label{lem:rank_Delta_DeltaT}
Let $\Delta(r, m)$ be the basis of $\RM(m-r-1,m)/\RM(r,m)$, then, $\Delta(r, m)\Delta(r, m)^\T$ has full rank.
\end{lemma}
\begin{proof}
  We note that $\Delta(r, m)$ corresponds to the basis vectors of the quotient space $\mathrm{RM}(m-r-1,m)/\mathrm{RM}(r,m)$. Thus, $\Delta(r,m)$ comprises of $\displaystyle \sum_{i=r+1}^{m-r-1}\begin{pmatrix}m\\ i\end{pmatrix} $ linearly independent rows, and is not self-dual containing \cite{RankRef}. Let us consider the rows of $\Delta(r, m)$ to be $\bar{d}_1$, \dots, $\bar{d}_l$, where $l = \underset{i=r+1}{\overset{m-r-1}{\sum}} {m \choose i}$. If the matrix $\Delta(r, m)\Delta(r, m)^\T$ is not a full rank matrix, then there exist $g$ rows that are linearly dependent for some $g \leq l$. Let the $i_1$, $\dots$, $i_g$ rows be linearly dependent. We note that the $(s,t)^{\mathrm{th}}$ element of $\Delta(r, m)\Delta(r, m)^\T$ is $\bar{d}_{s}\bar{d}_{t}^{\mathrm{T}}$. Then, there exists $b_m \in \mathbb{F}_2$, for $m=1,2,\dots,g$, where at least one $b_m$ is non-zero such that, for all $j \in \{1,\dots,l\}$,
  \begin{align}
        &\sum_{m=1}^{g} b_m (\Delta(r, m)\Delta(r, m)^\T)_{(i_m,j)} = 0,\nonumber\\
      \Rightarrow&\sum_{m=1}^{g} b_m \bar{d}_{i_m}\bar{d}_{j}^{\mathrm{T}} = 0~\Rightarrow\bar{d}_{j}\sum_{m=1}^{g} b_m \bar{d}_{i_m}^{\mathrm{T}} = 0,\nonumber\\
      \Rightarrow&\Delta(r, m)\sum_{m=1}^{g} b_m \bar{d}_{i_m}^{\mathrm{T}} = 0.
  \end{align}
  We note that $\Delta(r, m)\sum_{m=1}^{g} b_m \bar{d}_{i_m}^{\mathrm{T}} = 0$ means that $\sum_{m=1}^{g} b_m \bar{d}_{i_m}^{\mathrm{T}}$ is in the null space of $\Delta(r, m)$. Let $H^{\perp}(r,m)$ be the generator matrix of the $\RM(r,m)$ code. As $H^{\perp}(r,m)\Delta(r,m)^{\T} = 0$ ($\because \RM(m-r-1,m)^{\perp} = \RM(r,m)$), we obtain that $\sum_{m=1}^{g} b_m \bar{d}_{i_m}^{\mathrm{T}}$ is in the null space of $H(r,m) =\begin{bmatrix} H^{\perp}(r,m)\\\Delta(r,m)  \end{bmatrix}$. Thus, $\sum_{m=1}^{g} b_m \bar{d}_{i_m}^{\mathrm{T}} \in \mathrm{RM}(r,m)$, which is a contradiction because $\sum_{m=1}^{g} b_m \bar{d}_{i_m}^{\mathrm{T}} \in \mathrm{RM}(m-r-1,m)/\mathrm{RM}(r,m) \subset \mathrm{RM}(m-r-1,m)\setminus\mathrm{RM}(r,m)$. Thus, $\Delta(r, m)\Delta(r, m)^\T$ has full rank.
\end{proof}

\begin{theorem}\label{thm:EARM_n_e}
 The entanglement-assisted Reed-Muller code obtained using CSS construction from the classical Reed-Muller code $\mathrm{RM}(r,m)$ with parity check matrix $H$ requires $n_e$ pre-shared entangled qubits, where
\begin{align*}
    n_e =  \mathrm{gfrank}(HH^{\mathrm{T}}) = \underset{i=r+1}{\overset{m-r-1}{\sum}} {m \choose i}.
\end{align*}
\end{theorem}

\begin{proof}
For RM codes, we have $\mathrm{RM}(r_1, m) \subset \mathrm{RM}(r_2, m)$ for $r_1 < r_2$, and $\mathrm{RM}(r, m)^{\perp} = \mathrm{RM}(m-r-1, m)$. Thus, the rows of the parity check matrix $H(r, m)$ consists of the codewords of $\mathrm{RM}(m-r-1, m)$ and has dimensions $\sum_{i = 0}^{m-r-1} {m \choose i}$ {$\times 2^m$}.

When $r < \frac{m-1}{2}$, the parity check matrix of $\mathrm{RM}(r, m)$ can be written as
\begin{equation}
    H(r, m) = \begin{pmatrix} H^\perp (r, m) \\ \Delta (r, m)\end{pmatrix} = \begin{pmatrix}
        \RM(r, m) \\ \Delta (r, m)
    \end{pmatrix},
\end{equation}
where $\Delta (r, m) \equiv \RM(m-r-1) / \RM(r, m)$.

\noindent We note that $\RM(r, m)\RM(r, m)^\T = 0$ and $\RM(r, m)\Delta(r, m)^\T = 0$ as $\RM(m-r-1, m)$ is a dual-containing code. Thus, we obtain
\begin{align}
    &~~H(r, m) H(r, m)^\T \nonumber\\
    &~~= \begin{pmatrix}
        \RM(r, m)\RM(r, m)^\T & \RM(r, m)\Delta(r, m)^\T \\
        \Delta(r, m)\RM(r, m)^\T & \Delta(r, m)\Delta(r, m)^\T
    \end{pmatrix}\nonumber\\
    &~~= \begin{pmatrix}
        0 & 0 \\ 0 & \Delta(r, m)\Delta(r, m)^\T
    \end{pmatrix}\nonumber\\
    &\Rightarrow
    \mathrm{gfrank}(H(r,m)H(r,m)^{\mathrm{T}})\nonumber\\
    &~~~~= \mathrm{gfrank}(\Delta(r,m)\Delta(r,m)^{\mathrm{T}})\nonumber\\
    &~~~~= \underset{i=r+1}{\overset{m-r-1}{\sum}} {m \choose i},
\end{align}
because $\Delta(r,m)\Delta(r,m)^{\mathrm{T}}$ is a full rank matrix as shown in Lemma \ref{lem:rank_Delta_DeltaT}.
\end{proof}

\subsection{Zero-Coding Rate and Negative Catalytic Rate}
We note that we consider $\frac{m-1}{2} > r$ as we are computing the rate of the EA RM code.
We next show that the EA RM code has zero coding rate and a negative catalytic rate.

\begin{lemma} \label{lem:EARM_Zero_rate}
The EA RM codes constructed from $\mathrm{RM}(r,m)$ codes using the CSS construction have zero coding rate and negative catalytic rate.
\end{lemma}
\begin{proof}
Consider an $\mathrm{RM}(r,m)$ code with code parameters $[2^m, \underset{i=0}{\overset{r}{\sum}} {m \choose i}, 2^{m-r}]$ and {$\rho = n-k = \underset{j=0}{\overset{m-r-1}{\sum}} {m \choose j}$}. We first compute the number of logical qubits of the EA RM code constructed from the $\mathrm{RM}(r,m)$ code to be
\begin{align*}
&2k-n+n_e = n -2\rho + n_e,\nonumber\\
&=  2^m - 2\underset{i=0}{\overset{m-r-1}{\sum}} {m \choose i} + \underset{i=r+1}{\overset{m-r-1}{\sum}} {m \choose i},\nonumber\\
& =  \underset{i=0}{\overset{m}{\sum}}\! {m \choose i} \!\!-\!\!\!\!\! \underset{i=0}{\overset{m-r-1}{\sum}} \!{m \choose i} \!\!-\!\!\!\! \underset{i=0}{\overset{m-r-1}{\sum}}\! {m \choose m-i}\!\! +\!\!\!\! \underset{i=r+1}{\overset{m-r-1}{\sum}} \!{m \choose i},\nonumber\\ &~~~~~~\left(\because ~ 2^m = \underset{i=0}{\overset{m}{\sum}} {m \choose i}\text{ and }{m \choose m-i}={m \choose i} \right)   \nonumber\\
& =  \!\left(\underset{i=0}{\overset{m}{\sum}}\! {m \choose i} \!\!-\!\!\!\! \underset{i=0}{\overset{m-r-1}{\sum}}\! {m \choose i}\!\!\right) \!\!-\!\! \left(\underset{j=r+1}{\overset{m}{\sum}}\! {m \choose j} \!\!-\!\!\!\! \underset{i=r+1}{\overset{m-r-1}{\sum}}\! {m \choose i}\!\!\right),\nonumber\\
&~~~~~\text{ where }j=m-i, \nonumber\\
&= \left(\underset{i=m-r}{\overset{m}{\sum}} {m \choose i}\right) - \left(\underset{j=m-r}{\overset{m}{\sum}} {m \choose j}\right) = 0.
\end{align*}

We observe that EA RM codes have a zero rate. Thus, the catalytic rate is negative as $n_e$ is at least 1 for an EA code.
\end{proof}

\begin{example} \label{ex:EARM_code_example}
Using the $\mathrm{RM}(1,4)$ code provided in Example \ref{ex:RM_code_example} with the EA RM code construction provided in this Section, we obtain the EA RM code with the check matrix
\begin{align}
\mathcal{H} = \left[\begin{array}{cc|cc} H & H_{ex} & \mathbf{0} & \mathbf{0} \\ \mathbf{0} & \mathbf{0} & H & H_{ez}\end{array}\right],
\end{align}
where, 
\begin{align}
\!\!\!H\! \! 
&=\!\!\!\begin{bmatrix}\mathrm{Eval}^{(4)}(1)\\\mathrm{Eval}^{(4)}(x_1)\\\mathrm{Eval}^{(4)}(x_2)\\\mathrm{Eval}^{(4)}(x_3)\\\mathrm{Eval}^{(4)}(x_4)\\\mathrm{Eval}^{(4)}(x_1x_2)\\\mathrm{Eval}^{(4)}(x_1x_3)\\\mathrm{Eval}^{(4)}(x_1x_4)\\\mathrm{Eval}^{(4)}(x_2x_3)\\\mathrm{Eval}^{(4)}(x_2x_4)\\\mathrm{Eval}^{(4)}(x_3x_4)\end{bmatrix},\nonumber
\end{align}
\begin{align}
&=\!\!\!\left[\!\begin{array}{p{0.05cm}p{0.05cm}p{0.05cm}p{0.05cm}p{0.05cm}p{0.05cm}p{0.05cm}p{0.05cm}p{0.05cm}p{0.05cm}p{0.05cm}p{0.05cm}p{0.05cm}p{0.05cm}p{0.05cm}p{0.05cm}}
1 & 1 & 1 & 1 & 1 & 1 & 1 & 1 & 1 & 1 & 1 & 1 & 1 & 1 & 1 & 1\\
0 & 0 & 0 & 0 & 0 & 0 & 0 & 0 & 1 & 1 & 1 & 1 & 1 & 1 & 1 & 1\\
0 & 0 & 0 & 0 & 1 & 1 & 1 & 1 & 0 & 0 & 0 & 0 & 1 & 1 & 1 & 1\\
0 & 0 & 1 & 1 & 0 & 0 & 1 & 1 & 0 & 0 & 1 & 1 & 0 & 0 & 1 & 1\\
0 & 1 & 0 & 1 & 0 & 1 & 0 & 1 & 0 & 1 & 0 & 1 & 0 & 1 & 0 & 1\\
0 & 0 & 0 & 0 & 0 & 0 & 0 & 0 & 0 & 0 & 0 & 0 & 1 & 1 & 1 & 1\\
0 & 0 & 0 & 0 & 0 & 0 & 0 & 0 & 0 & 0 & 1 & 1 & 0 & 0 & 1 & 1\\
0 & 0 & 0 & 0 & 0 & 0 & 0 & 0 & 0 & 1 & 0 & 1 & 0 & 1 & 0 & 1\\
0 & 0 & 0 & 0 & 0 & 0 & 1 & 1 & 0 & 0 & 0 & 0 & 0 & 0 & 1 & 1\\
0 & 0 & 0 & 0 & 0 & 1 & 0 & 1 & 0 & 0 & 0 & 0 & 0 & 1 & 0 & 1\\
0 & 0 & 0 & 1 & 0 & 0 & 0 & 1 & 0 & 0 & 0 & 1 & 0 & 0 & 0 & 1
\end{array}\right]\!\!,\label{eqn:RM_Example_H}
\end{align}
and $H_{ex}$ and $H_{ez}$ correspond to the operations of the stabilizer on the receiver end qubits of the pre-shared entangled qubits. We note that $H_{ex}$ and $H_{ez}$ are
\begin{align}
    H_{ex} &= \begin{bmatrix}    
    0 & 0 & 0 & 0 & 0 & 0\\
    0 & 0 & 0 & 0 & 0 & 0\\
    0 & 0 & 0 & 0 & 0 & 0\\
    0 & 0 & 0 & 0 & 0 & 0\\
    0 & 0 & 0 & 0 & 0 & 0\\
    1 & 0 & 0 & 0 & 0 & 0\\
    0 & 1 & 0 & 0 & 0 & 0\\
    0 & 0 & 1 & 0 & 0 & 0\\
    0 & 0 & 0 & 1 & 0 & 0\\
    0 & 0 & 0 & 0 & 1 & 0\\
    0 & 0 & 0 & 0 & 0 & 1\\
    \end{bmatrix}
    \text{ and }\\
    H_{ez} &= \begin{bmatrix}
    0 & 0 & 0 & 0 & 0 & 0\\
    0 & 0 & 0 & 0 & 0 & 0\\
    0 & 0 & 0 & 0 & 0 & 0\\
    0 & 0 & 0 & 0 & 0 & 0\\
    0 & 0 & 0 & 0 & 0 & 0\\
    0 & 0 & 0 & 0 & 0 & 1\\
    0 & 0 & 0 & 0 & 1 & 0\\
    0 & 0 & 0 & 1 & 0 & 0\\
    0 & 0 & 1 & 0 & 0 & 0\\
    0 & 1 & 0 & 0 & 0 & 0\\
    1 & 0 & 0 & 0 & 0 & 0\\
    \end{bmatrix}
\end{align}
The matrices $H_{ex}$ and $H_{ez}$ are obtained using the symplectic Gram-Schmidt orthogonalization procedure provided in the supplementary material of \cite{Brun06} and \cite{Nadkarni_NBEASC}. The EA RM code obtained from the $\mathrm{RM}(1,4)$ code has parameters $[[16, 0, \geq 8 ; 6]]$ quantum code that requires $6$ pre-shared entangled qubits.
\end{example}

\begin{remark}
 We note that the RM codes for constructing quantum codes using the EA construction have a smaller value of $r$ compared to those used to construct QRM codes. As the code distance of the $\mathrm{RM}(r,m)$ code is $2^{(m-r)}$, we obtain codes with better code distances using the RM codes that are not dual-containing codes. However, from Lemma \ref{lem:EARM_Zero_rate}, these RM codes provide zero rate quantum codes. Thus, we further explore techniques to construct non-zero rate codes from $\mathrm{RM}(r,m)$ codes with $r < (m-1)/2$.
\end{remark}

\begin{remark}\label{rem:EARMTPC_PositiveCodingRate}
From Theorem \ref{thm:EARM_n_e}, Lemma \ref{lem:EARM_Zero_rate}, and Corollary \ref{coro:CodingLawPossibility}, as $n_e = \sum_{i=r+1}^{m-r-1}{m \choose i} < \sum_{i=0}^{m-r-1}{m \choose i} = \rho$, the quantum{/EA RM} TPC constructed from the classical RM codes with $(m-1)/2 > r$ will have positive rate although the {QRM/EA RM} codes have zero coding rate. {We have provided an explicit proof for the same based on the code parameters in Appendix \ref{app:EARM_rate_Proof} for completeness.}
\end{remark}

\section{Entanglement-Assisted Tensor Product Codes from Reed-Muller Codes} \label{sec:EATPC_RM}

\subsection{Positive rate of the EA {RM} TPC}

We recall that the parameters of a TPC obtained from two classical codes $C_1[n_1, k_1,d_1]$ and $C_2[n_2,k_2,d_2]$ is $[n_1n_2, n_1n_2 - \rho_1\rho_2, \mathrm{min}(d_1,d_2)]$. Using two classical RM codes $\mathrm{RM}(r_1,m_1)$ and $\mathrm{RM}(r_2,m_2)$, namely, $C_1[2^{m_1}, \underset{i=0}{\overset{r_1}{\sum}} {m_1 \choose i}, 2^{m_1-r_1}]$ and $C_2[2^{m_2}, \underset{i=0}{\overset{r_2}{\sum}} {m_2 \choose i}, 2^{m_2-r_2}]$, with the TPC construction, we obtain a TPC $C$ with the code length $n = 2^{m_1+m_2}$, number of logical qubits, $k = 2^{m_1+m_2} - \left(\underset{i=0}{\overset{m_1-r_1-1}{\sum}} {m_1 \choose i}\right)\left(\underset{i=0}{\overset{m_2-r_2-1}{\sum}} {m_2 \choose i}\right)$, and minimum distance $d = \mathrm{min}\left(2^{m_1-r_1},2^{m_2-r_2}\right)$.
The EA CSS code obtained from {the TPC $C$} has the following parameters:
\begin{align*}
[[n, n-2\rho + n_e, \geq d, n_e]],
\end{align*}
 where
\begin{align*} 
 n_e\! &=\! \mathrm{gfrank}(HH^{\mathrm{T}}) \!= \mathrm{gfrank}(H_1H_1^{\mathrm{T}})\mathrm{gfrank}(H_2H_2^{\mathrm{T}}) \nonumber\\
 &~~~~~~~~~~~~~~~~~~~~~~~~~~~~~~~~~~~~~~~~~~~~~~[\text{From \cite{EABQTPC}}],\nonumber\\
 \rho &= \rho_1\rho_2 = \left(\underset{i=0}{\overset{m_1-r_1-1}{\sum}} {m_1 \choose i}\right)\left(\underset{i=0}{\overset{m_2-r_2-1}{\sum}} {m_2 \choose i}\right),\nonumber\\
 n_e &= n_{e1}n_{e2} = \left(\underset{i=r_1+1}{\overset{m_1-r_1-1}{\sum}} {m_1 \choose i}\!\right)\!\!\left(\underset{i=r_2+1}{\overset{m_2-r_2-1}{\sum}} {m_2 \choose i}\!\right).\nonumber
 \end{align*}

We next provide a result on the coding rate of an EA RM TPC.

\begin{theorem} \label{thm:EARMTPC_PositiveCodingRate}
The EA RM TPC obtained from the $\mathrm{RM}(r_1,m_1)$ and $\mathrm{RM}(r_2,m_2)$ classical RM codes has a positive coding rate.
\end{theorem}

\!\!\!\!\!{Theorem \ref{thm:EARMTPC_PositiveCodingRate} has been proved in Remark \ref{rem:EARMTPC_PositiveCodingRate}. However, we provide an explicit proof of Theorem \ref{thm:EARMTPC_PositiveCodingRate} in Appendix \ref{app:EARM_rate_Proof} as some results in the proof will be used in Section \ref{sec:CatalyticRate}.}

From Lemma \ref{lem:EARM_Zero_rate} and Theorem \ref{thm:EARMTPC_PositiveCodingRate}, we have demonstrated the coding analog of superadditivity by constructing positive rate EA RM TPCs from classical RM codes which generated zero rate EA RM codes.

\subsection{Catalytic Rate Computation}\label{sec:CatalyticRate}

The catalytic rate $r_c$ of the EA RM TPC is
\begin{align}
&\text{$r_c$ =  (number of logical qubits - }n_e)/n,\nonumber
\end{align}
From equation \eqref{eqn:EARMTPC_logical_qubits_no} in Appendix \ref{app:EARM_rate_Proof}, we obtain
\begin{align}
&  \text{number of logical qubits } - n_e\nonumber\\
 &~~= \left(\underset{j=0}{\overset{r_1}{\sum}} {m_1 \choose j}\right)\left(\underset{i=0}{\overset{m_2}{\sum}} {m_2 \choose i}-\underset{i=r_2+1}{\overset{m_2-r_2-1}{\sum}} {m_2 \choose i}\right)\nonumber\\
 &~~~~  - \left(\underset{i=r_1+1}{\overset{m_1-r_1-1}{\sum}} {m_1 \choose i}\right)\left(\underset{i=r_2+1}{\overset{m_2-r_2-1}{\sum}} {m_2 \choose i}\right),\nonumber\\
 &~~= \left(\underset{j=0}{\overset{r_1}{\sum}} {m_1 \choose j}\right)\left(\underset{i=0}{\overset{r_2}{\sum}} {m_2 \choose i}+\underset{i=m_2-r_2}{\overset{m_2}{\sum}} {m_2 \choose i}\right) \nonumber\\
 &~~~~ - \left(\underset{i=r_1+1}{\overset{m_1-r_1-1}{\sum}} {m_1 \choose i}\right)\left(\underset{i=r_2+1}{\overset{m_2-r_2-1}{\sum}} {m_2 \choose i}\right),\nonumber\\
  &~~= \left(\underset{j=0}{\overset{r_1}{\sum}} {m_1 \choose j}\right)\left(\underset{j=0}{\overset{r_2}{\sum}} {m_2 \choose j}+\underset{u=0}{\overset{r_2}{\sum}} {m_2 \choose u}\right) \nonumber\\
 &~~~~ - \left(\underset{i=r_1+1}{\overset{m_1-r_1-1}{\sum}} {m_1 \choose i}\right)\left(\underset{i=r_2+1}{\overset{m_2-r_2-1}{\sum}} {m_2 \choose i}\right),\nonumber\\
  &~~~~~~~\text{ where }u = m_2-i,\nonumber\\
    &~~= 2\left(\underset{j=0}{\overset{r_1}{\sum}} {m_1 \choose j}\right)\left(\underset{u=0}{\overset{r_2}{\sum}} {m_2 \choose u}\right) \nonumber\\
 &~~~~ - \left(\underset{i=r_1+1}{\overset{m_1-r_1-1}{\sum}} {m_1 \choose i}\right)\left(\underset{i=r_2+1}{\overset{m_2-r_2-1}{\sum}} {m_2 \choose i}\right),\label{eqn:catalytic_rate}
\end{align}
which could be positive for certain values of $r_1$, $r_2$, $m_1$, and $m_2$. {We note that $\text{number of logical qubits } - n_e = 2(R_1+R_2-R_1R_2)-1$, where $R_i$s are the coding rates of the classical component codes of the TPC.}

In Section \ref{sec:tpcexample}, we provide a subclass of EA RM TPCs which have positive catalytic rates.

\subsection{EA {RM TPCs} with Positive Catalytic Rate}\label{sec:tpcexample}

We next provide a subclass of EA {RM} TPCs that have positive catalytic rates.

\begin{theorem} \label{thm:EARMTPC_Subclass_PositiveCatalyticRate}
 For $r, s \in \mathbb{Z}^+$, the EA {RM} TPCs constructed from classical $\mathrm{RM}(r,2r+s)$ codes have positive catalytic rates if $s \leq l(r)$ and have zero or negative catalytic rate if $s > l(r)$, where 
 \begin{align}
 l(r) \!\!=& \mathrm{max}~i\nonumber\\
 &~\text{subject to } \displaystyle\sum_{u=0}^{r}\!\!\binom{2r+i}{u}\! >\!\! \frac{2^{2r+i}}{2+\sqrt{2}} \text{ and }i\! \in\!\mathbb{Z}^+\!.\nonumber
 \end{align}
\end{theorem}

\begin{proof}
The classical $\mathrm{RM}(r,2r+s)$ codes are not dual-containing codes for $s \geq 2$ as $(m-1)/2 > r$. Thus, from Lemma \ref{lem:EARM_Zero_rate}, the EA RM codes constructed from $\mathrm{RM}(r,2r+s)$ codes have zero coding rate and negative catalytic rate. From equation \eqref{eqn:catalytic_rate}, for the EA RM TPCs constructed by considering both the classical component codes as $\mathrm{RM}(r,2r+s)$, we obtain
\begin{align}
&n_{\text{catalytic}}\nonumber\\
&=\text{number of logical qubits }-n_e\nonumber\\
&= 2\left(\underset{u=0}{\overset{r_1}{\sum}} {m_1 \choose u}\right)\left(\underset{v=0}{\overset{r_2}{\sum}} {m_2 \choose v}\right) \nonumber\\
& - \left(\underset{w=r_1+1}{\overset{m_1-r_1-1}{\sum}} {m_1 \choose w}\right)\left(\underset{x=r_2+1}{\overset{m_2-r_2-1}{\sum}} {m_2 \choose x}\right),\nonumber\\ 
&= 2\left(\underset{u=0}{\overset{r}{\sum}} {2r+s \choose u}\right)^2 \!\!\!-\!\! \left(\underset{w=r+1}{\overset{r+s-1}{\sum}} {2r+s \choose w}\right)^2\!\!\!. \label{eqn:ex_catalytic_rate_1}
\end{align}
 For the catalytic rate to be positive, $n_{\text{catalytic}}$ should be greater than $0$. Thus, from equation \eqref{eqn:ex_catalytic_rate_1}, we want $2\left(\underset{u=0}{\overset{r}{\sum}} {2r+s \choose u}\right)^2 - \left(\underset{w=r+1}{\overset{r+s-1}{\sum}} {2r+s \choose w}\right)^2 > 0$ for the EA RM TPC to have positive catalytic rate. We note that for $a, b \in \mathbb{Z}^+$, $2a^2 - b^2 >0 \iff \sqrt{2}a > b$. Thus, choosing $a = \underset{u=0}{\overset{r}{\sum}} {2r+s \choose u}$ and $ b = \underset{w=r+1}{\overset{r+s-1}{\sum}} {2r+s \choose w}$, we obtain that the EA RM TPC has a positive catalytic rate if and only if
 \begin{align}
     &\sqrt{2}\left(\underset{u=0}{\overset{r}{\sum}} {2r+s \choose u}\right) \!>\! \left(\underset{w=r+1}{\overset{r+s-1}{\sum}} {2r+s \choose w}\right),\nonumber\\
     \Rightarrow &\sqrt{2}\left(\underset{u=0}{\overset{r}{\sum}} {2r+s \choose u}\!\right) \!>\! \left(2^{2r+s} \!- 2\underset{w=0}{\overset{r}{\sum}} {2r+s \choose w}\!\right),\nonumber\\
     \Rightarrow &\underset{u=0}{\overset{r}{\sum}} {2r+s \choose u} \!>\! \frac{2^{2r+s}}{2+\sqrt{2}}.\nonumber
 \end{align}

 We next show that if $\underset{u=0}{\overset{r}{\sum}} {2r+s \choose u} < \frac{2^{2r+s}}{2+\sqrt{2}}$ then $\underset{u=0}{\overset{r}{\sum}} {2r+s+1 \choose u} < \frac{2^{2r+s+1}}{2+\sqrt{2}}$. We first consider that $\underset{u=0}{\overset{r}{\sum}} {2r+s \choose u} < \frac{2^{2r+s}}{2+\sqrt{2}}$, i.e., $(2+\sqrt{2})\underset{u=0}{\overset{r}{\sum}} {2r+s \choose u} < 2^{2r+s}$, then we obtain
 \begin{align}
     &(2+\sqrt{2})\underset{u=0}{\overset{r}{\sum}} {2r+s+1 \choose u} -2^{2r+s+1}\nonumber\\
     &= (2+\sqrt{2})\left(1\!+\!\underset{u=1}{\overset{r}{\sum}} {2r+s+1 \choose u}\!\right)\!-\! 2^{2r+s+1}\label{eqn:eqn_1}\\
    &= (2\!+\!\sqrt{2})\!\left(1\!+\!\underset{u=1}{\overset{r}{\sum}} {2r\!+\!s \choose u}\! +\! {2r\!+\!s \choose u\!-\!1}\!\right)\!\!-\! 2^{2r+s+1}\nonumber\\
    & ~~\left(\because\text{By Pascal's rule }\!,\binom{n}{k} \!\!=\!\! \binom{n\!-\!1}{k}\! +\! \binom{n\!-\!1}{k\!-\!1}\!\right)\nonumber\\
    &= (2\!+\!\sqrt{2})\left(\!{2r\!+\!s \choose 0}\!+\!\underset{u=1}{\overset{r}{\sum}} \!{2r\!+\!s \choose u}\! +\!\! \underset{v=0}{\overset{r-1}{\sum}} \!{2r\!+\!s \choose v}\!\right)\nonumber\\
    &~~~~- 2^{2r+s+1}, \text{ where }v = u-1,\nonumber\\
    &= (2\!+\!\sqrt{2})\!\left(\underset{u=0}{\overset{r}{\sum}}\! {2r\!+\!s \choose u} \!\!+\! \underset{v=0}{\overset{r-1}{\sum}}\! {2r\!+\!s \choose v}\!\right)\!\!-\! 2^{2r+s+1},\nonumber\\
    &= \underbrace{\left(\left((2+\sqrt{2})\underset{u=0}{\overset{r}{\sum}} {2r+s \choose u}\right) - 2^{2r+s}\right)}_{A} \nonumber\\
    &~~~~~~~+ \underbrace{\left(\left((2+\sqrt{2})\underset{v=0}{\overset{r-1}{\sum}} {2r+s \choose v}\right)- 2^{2r+s}\right)}_{B} \!<\! 0,
\end{align}
as both the terms A and B are less than 0 as $\left((2+\sqrt{2})\underset{u=0}{\overset{r}{\sum}} {2r+s \choose u}\right) < 2^{2r+s}$ and $\left((2+\sqrt{2})\underset{v=0}{\overset{r-1}{\sum}} {2r+s \choose v}\right) < \left((2+\sqrt{2})\underset{u=0}{\overset{r}{\sum}} {2r+s \choose u}\right)$. Thus, if any EA RM TPC based on $\mathrm{RM}(r,2r+s)$ code has zero or negative catalytic rate, then all the EA RM TPCs based on $\mathrm{RM}(r,2r+s +l)$ codes have negative catalytic rates for $l \in \mathbb{Z}^+$. 

We define $l(r)$ as follows:
 \begin{align}
 \!\!l(r) \!\!=& \mathrm{max}~i\nonumber\\
 &~\text{subject to } \displaystyle\sum_{u=0}^{r}\!\!\binom{2r\!+\!i}{u} \!>\!\! \frac{2^{2r+i}}{2\!+\!\sqrt{2}}\text{ and } i \in\mathbb{Z}^+.\nonumber
 \end{align}
Thus, the EA RM TPCs obtained from $\mathrm{RM}(r,m)$ codes where $2r+ 2\leq m \leq 2r + l(r)$ have positive catalytic rates, while the EA RM TPCs obtained from $\mathrm{RM}(r,m)$ codes where $m > 2r + l(r)$ have zero or negative catalytic rate.
\end{proof}

We need to solve the integer optimization problem provided in Theorem \ref{thm:EARMTPC_Subclass_PositiveCatalyticRate} to obtain $l(r)$. Since the binomial coefficient ${2r+i \choose u} < \frac{(2r+i)^u}{u!}$, the original partial sum of binomial coefficients can be regarded as a partial sum of an exponential series. Thus, $\displaystyle\sum_{u=0}^{r}\!\!\binom{2r+i}{u} < \displaystyle\sum_{u=0}^{r}\!\!\frac{(2r+i)^u}{u!}$. We can solve the integer optimization problem in Theorem \ref{thm:EARMTPC_Subclass_PositiveCatalyticRate} by taking derivatives and examining the saddle points of the polynomial  $\displaystyle\sum_{u=0}^{r}\!\!\frac{(2r+i)^u}{u!}$  for a given $r$ and checking if the positive integer values of $i$ near the optimal value satisfy $\displaystyle\sum_{u=0}^{r}\!\!\binom{2r+i}{u}\! >\!\! \frac{2^{2r+i}}{2+\sqrt{2}}$. Indeed, this optimization is as complex as solving for the roots of the polynomial equation, motivating the difficulty in the pursuit of an exact analytical solution.

We note that we are practically not interested in $\mathrm{RM}$ codes with $m > 300$ as the length of the classical $\mathrm{RM}$ code is $2^m$ and as the number of atoms in this universe is much less than $2^{300}$. In Table \ref{tab:l_r}, we list the values of $l(r)$ computed using numerical analysis for various ranges of $r$. We note that for all practical values of $m$, i.e., $m <= 300$, $l(r) <= 10$.

\begin{table}[ht]
    \centering
    \caption{Table of $l(r)$ for various ranges of $r$}
    \begin{tabular}{|c|c|}
        \hline
        \hline
        r & l(r)\\
        \hline
        \hline
        r $\in$ [1,5] & 2\\
        \hline
        r $\in$ [6,13] & 3\\
        \hline
        r $\in$ [14,24] & 4\\
        \hline
        r $\in$ [25,39] & 5\\
        \hline
        r $\in$ [40,57] & 6\\
        \hline
        r $\in$ [58,78] & 7\\
        \hline
        r $\in$ [79,103] & 8\\
        \hline
        r $\in$ [104,131] & 9\\
        \hline
        r $\in$ [132,162] & 10\\
        \hline
        \hline
    \end{tabular}
    \label{tab:l_r}
\end{table}

{We note that the condition in Theorem \ref{thm:EARMTPC_Subclass_PositiveCatalyticRate} is special case of the condition in Corollary \ref{coro:EACSS_positive_catalytic} tailored to the EA RM TPCs. It can be viewed as a simplification of the condition in Theorem \ref{thm:EARMTPC_Subclass_PositiveCatalyticRate} to obtain a condition based on the RM code parameters $r$ and $m = 2r+s$.}

We next show explicitly we always obtain a positive catalytic rate EA RM TPC from classical $\mathrm{RM}(i,2i+2)$ codes.
\begin{lemma} \label{lem:EARMTPC_2i+2}
 For $i \in \mathbb{Z}^+$, the EA {RM} TPCs constructed from classical $\mathrm{RM}(i,2i+2)$ codes have positive catalytic rates.
\end{lemma}
\begin{proof}
Substituting $s=2$ in equation \eqref{eqn:ex_catalytic_rate_1}, we obtain the difference between the number of logical qubits and $n_e$ to be 
\begin{align}
    &\text{number of logical qubits }-n_e\nonumber\\
    &= 2\left(\underset{u=0}{\overset{r}{\sum}} {2r+2 \choose u}\right)^2 - {2r+2 \choose r+1}^2. \label{eqn:n_c_s2}
\end{align}
Thus, using $r=1$ in equation \eqref{eqn:n_c_s2}, we obtain
\begin{align}
    &\text{number of logical qubits }-n_e\nonumber\\
&= 2\left(\underset{j=0}{\overset{1}{\sum}} {4 \choose j}\right)^2 - {4 \choose 2}^2 = 14 >0. \label{eqn:ex_catalytic_rate_2}
\end{align}
We note that, for $r>1$ and $r \in \mathbb{Z}^+$, 
\begin{align}
    &{2r+2 \choose r} + {2r+2 \choose r-1} - {2r+2 \choose r+1} \nonumber\\
    &= \frac{(2r\!+\!2)!}{r!(r\!+\!2)!} \!+\! \frac{(2r\!+\!2)!}{(r\!-\!1)!(r\!+\!3)!} \!-\! \frac{(2r\!+\!2)!}{(r\!+\!1)!(r\!+\!1)!},\nonumber\\
    &= \frac{(2r\!+\!2)!}{(r\!+\!1)!(r\!+\!3)!}\big( (r^2\!+\!4r\!+\!3) + (r^2 \!+\! r) \nonumber\\
    &\hspace{5.5cm} \!-\! (r^2 \!+\! 5r\! +\! 6)\big),\nonumber\\
    &= \frac{(2r\!+\!2)!}{(r\!+\!1)!(r\!+\!3)!}\!(r^2\!-\!3) \!>\!0,\nonumber\\
    &\Rightarrow \!\!{2r\!+\!2 \choose r}\!\! +\!\! {2r\!+\!2 \choose r\!-\!1}\!\! >\!\! {2r\!+\!2 \choose r\!+\!1},\nonumber\\
    &\Rightarrow \!\underset{u=0}{\overset{r}{\sum}} {2r\!+\!2 \choose u} \!>\! {2r\!+\!2 \choose r\!+\!1},\text{ for }r\!>\!1\text{ and }r \!\in\! \mathbb{Z}^+, \label{eqn:ex_catalytic_rate_3}
\end{align}
as the summation $\left(\underset{u=0}{\overset{r}{\sum}} {2r+2 \choose u}\right)$ contains ${2r+2 \choose r} + {2r+2 \choose r-1}$. 
Thus, substituting equation \eqref{eqn:ex_catalytic_rate_3} in equation \eqref{eqn:n_c_s2}, we obtain
\begin{align}
&\text{number of logical qubits }-n_e\nonumber\\     &=2\left(\underset{u=0}{\overset{r}{\sum}} {2r+2 \choose u}\right)^2 - {2r+2 \choose r+1}^2 > 0.,\nonumber
\end{align}
Thus, for $r \in \mathbb{Z}^+$, from equations  \eqref{eqn:ex_catalytic_rate_2} and \eqref{eqn:ex_catalytic_rate_3}, we obtain
\begin{align}
    &\text{number of logical qubits }-n_e >0,\nonumber\\
    \Rightarrow&\text{Catalytic rate }\nonumber\\
    &= (\text{number of logical qubits }-n_e)/n >0.\nonumber\qedhere
\end{align}
\end{proof}

\begin{example}
In Example \ref{ex:EARM_code_example}, we showed that the EA RM code constructed from the $\mathrm{RM}(1,4)$ code discussed in Example \ref{ex:RM_code_example} is a zero-rate $[[16, 0, \geq 8; 6]]$ quantum code. We construct the EA RM TPC by considering both classical component codes to be $\mathrm{RM}(1,4)$ codes. The classical RM TPC with both component codes being $\mathrm{RM}(1,4)$ codes has code parameters $[256, 135, 8]$. The EA RM TPC constructed from the classical RM TPC is a quantum code that requires $36$ pre-shared entangled qubits and is a $[[256, 50, \geq 8; 36]]$ code. Thus, the EA rate and the catalytic rates of the code are $(n+n_e-\rho)/n = 50/256 = 0.1953$ and $((n+n_e-\rho)-n_e)/n = (50-36)/256 = 0.05468$, which are both positive. 
\end{example}

Based on Lemma \ref{lem:EARMTPC_2i+2}, we provide some examples of codes with positive catalytic rates with both component codes of the TPC being the same RM codes in Table \ref{tab:examples}.

\begin{table*}[ht]
\caption{Examples of RM codes that provide positive catalytic rate EA RM TPCs}
\label{tab:examples}
\centering
\resizebox{\textwidth}{!}{\begin{tabular}{|c|c|c|c|}
\hline
{\begin{tabular}{c} $\RM(r, m)$ code parameters\\ $\left[ 2^m , \underset{i=0}{\overset{r}{\sum}} {m \choose i} , 2^{m-r} \right]$ \end{tabular}}& {\begin{tabular}{c} EA RM parameters\\$\left[\left[ 2^m , 0 ,\geq 2^{m-r} , \underset{i=r+1}{\overset{m-r-1}{\sum}}{m \choose i}\right]\right]$\end{tabular}} & {\begin{tabular}{c} EA RM TPC parameters\\$\left[\left[2^{2m} , 2\left(\underset{i=0}{\overset{r}{\sum}}{m \choose i}\right)^2 ,\geq 2^{m-r}, \left(\underset{i=r+1}{\overset{m-r-1}{\sum}}{m \choose i}\right)^2 \right]\right]$\end{tabular}} & {\begin{tabular}{c} EA RM TPC catalytic rate\\$2^{-2m}\left(2\left(\underset{i=0}{\overset{r}{\sum}}{m \choose i}\right)^2 - \left(\underset{i=r+1}{\overset{m-r-1}{\sum}}{m \choose i}\right)^2\right)$\end{tabular}}\\
\hline
$\mathrm{RM}(1,4),~ C[ 16 , 5 , 8 ]$ & $[[ 16 , 0 ,\geq 8 , 6 ]]$ & $[[ 256 , 50 ,\geq 8 , 36 ]]$ & $0.0546875$\\
\hline
$\mathrm{RM}(2,6),~ C[ 64 , 22 , 16 ]$ & $[[ 64 , 0 ,\geq 16 , 20 ]]$ & $[[ 4096 , 968 ,\geq 16 , 400 ]]$ & $0.138671875$\\
\hline
$\mathrm{RM}(3,8),~ C[ 256 , 93 , 32 ]$ & $[[ 256 , 0 ,\geq 32 , 70 ]]$ & $[[ 65536 , 17298 ,\geq 32 , 4900 ]]$ & $0.189178466796875$\\
\hline
$\mathrm{RM}(4,10),~ C[ 1024 , 386 , 64 ]$ & $[[ 1024 , 0 ,\geq 64 , 252 ]]$ & $[[ 1048576 , 297992 ,\geq 64 , 63504 ]]$ & $0.22362518310546875$\\
\hline
$\mathrm{RM}(5,12),~ C[ 4096 , 1586 , 128 ]$ & $[[ 4096 , 0 ,\geq 128 , 924 ]]$ & $[[ 16777216 , 5030792 ,\geq 128 , 853776 ]]$ & $0.24896955490112305$\\
\hline
\end{tabular}}
\end{table*}

Based on the relation of the parameters $r$ and $m$ of the classical RM code, we can construct either QRM codes using the CSS construction or positive rate EA RM TPCs using the construction presented in this Section. Based on Theorems \ref{thm:EARMTPC_PositiveCodingRate} and \ref{thm:EARMTPC_Subclass_PositiveCatalyticRate}, in Figure \ref{fig:m-r-plane}, we illustrate the partitioning of the $m-r$ plane into various regions for which the quantum RM TPCs constructed from {TPCs based on} classical $\RM(r,m)$ codes are either QRM codes, positive catalytic rate EA RM TPCs, or non-positive catalytic rate EA RM TPCs.


\begin{figure}[t]
    \centering
    \includegraphics[width = \columnwidth]{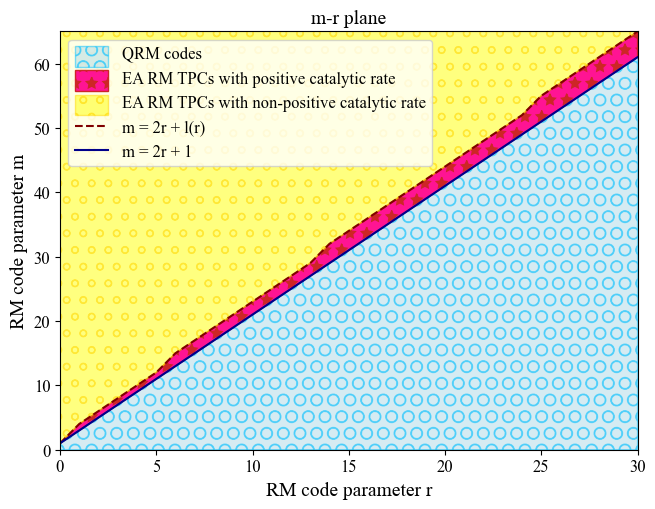}
    \caption{The $m-r$ plane is partitioned into three regions, where the regions with the pattern O correspond to $\mathrm{RM}$ codes that can be used to construct entanglement-unassisted quantum RM codes. These codes satisfy $r \geq (m-1)/2$. EA RM TPCs with positive and non-positive catalytic rates can be constructed using the $\mathrm{RM}(r,m)$ codes based on the region with pattern $\star$ and o, respectively. The line between the $\star$-patterned and o-patterned region corresponds to $m=2r + l(r)$ and the line between the $\star$-patterned and {O}-patterned region corresponds to $m = 2r+1$.}
    \label{fig:m-r-plane}
\end{figure}

{In Table \ref{tab:code_params}, we have summarized the parameters of various classical and quantum codes described in this paper.}

\begin{table*}[ht]
\caption{{Code parameters of various classical and quantum codes}}
\label{tab:code_params}
\centering
{
\resizebox{0.5\textwidth}{!}{\begin{tabular}{|c|c|}
\hline
Classical code & Code parameters\\
\hline
$\RM(r_i, m_i)$ code, $i \in \{1,2\}$ & $\left[ 2^{m_i} , \underset{j=0}{\overset{r_i}{\sum}} {m_i \choose j} , 2^{m_i-r_i} \right]$\\
\hline
RM TPC from $\RM(r_i, m_i)$ codes, $i \in \{1,2\}$ & $\left[ 2^{m_i} , \underset{j=0}{\overset{r_i}{\sum}} {m_i \choose j} , 2^{m_i-r_i} \right]$\\
\hline
\end{tabular}}}\vspace{0.3cm}
{
\resizebox{\textwidth}{!}{\begin{tabular}{|c|c|c|}
\hline
Quantum code & Code parameters & Catalytic rate\\
\hline
EA $\RM(r_i, m_i)$ code, $i \in \{1,2\}$ & $\left[\left[ 2^{m_i} , 0 ,\geq 2^{m_i-r_i}; \underset{j=r_i+1}{\overset{m_i-r_i-1}{\sum}}{m_i \choose j}\right]\right]$ & -$2^{-m_i}\left(\underset{j=r_i+1}{\overset{m_i-r_i-1}{\sum}}{m_i \choose j}\right) $ \\
\hline
\begin{tabular}{c} EA RM TPC code \\ from $\RM(r_i, m_i)$ codes, $i \in \{1,2\}$ \end{tabular} & $\left[\left[2^{2m} , 2\left(\underset{j=0}{\overset{r_1}{\sum}} {m_1 \choose j}\right)\left(\underset{u=0}{\overset{r_2}{\sum}} {m_2 \choose u}\right) ,\geq 2^{m-r}, \left(\underset{i=r_1+1}{\overset{m_1-r_1-1}{\sum}} {m_1 \choose i}\right)\left(\underset{i=r_2+1}{\overset{m_2-r_2-1}{\sum}} {m_2 \choose i}\right) \right]\right]$ & $2^{-2m}\left(2\left(\underset{j=0}{\overset{r_1}{\sum}} {m_1 \choose j}\right)\left(\underset{u=0}{\overset{r_2}{\sum}} {m_2 \choose u}\right) - \left(\underset{i=r_1+1}{\overset{m_1-r_1-1}{\sum}} {m_1 \choose i}\right)\left(\underset{i=r_2+1}{\overset{m_2-r_2-1}{\sum}} {m_2 \choose i}\right)\right)$\\
\hline
\end{tabular}}}
\end{table*}

\begin{remark}
 In Theorem 5 of \cite{EntropicSingletonBound}, Grassl \textit{et al.} have shown that the net entanglement production is at most $\mathrm{max}(0,n-2d+2)$. Thus, the catalytic rate is always negative when $n-2d+2 < 0$. For EA RM TPCs with both classical component codes being $\RM(r,m)$, the length $n = 2^{2m} = 2^{4r+2s}$, $d \geq 2^{(m-r)} = 2^{(2r+s-r)} = 2^{(r+s)}$, where $s \leq l(r)$. Thus, $n-2d+2 \leq 2^{(4r+2s)} - 2^{(r+s+1)} + 2$. As $4r+2s \geq r+s+1$ for $r,s \in \mathbb{Z}^+$, we obtain $2^{(4r+2s)} - 2^{(r+s+1)} + 2 > 0$. Thus, the EA RM TPCs proposed in Theorem \ref{thm:EARMTPC_Subclass_PositiveCatalyticRate} having positive catalytic rate do not violate the Grassl \textit{et al.}'s bound provided in Theorem 5 of \cite{EntropicSingletonBound} when $d = 2^{(m-r)} = 2^{(r+s)}$. However, the exact distance of the EA RM TPCs needs to be computed towards a more general statement.
\end{remark}

\begin{remark}
The construction of EA RM codes over qudits using classical non-binary RM codes is a natural extension of this work. {As classical non-binary RM codes have a non-zero coding rate, from Corollary \ref{coro:CodingLawPossibility},} the EA RM codes over qudits have positive catalytic rates. {We note that} EA Reed-Solomon (RS) codes\cite{EARS_codes_conf}\cite{EARS_codes}, which are a special case of the EA RM codes over qudits, have been shown to have positive catalytic rate. {Obtaining the exact expression for $l(r)$ is beyond the scope of this paper. We refer an interested reader to \cite{CodingAnalogSuperadditivity}, \cite{Nadkarni_NBEASC}, \cite{EARS_codes}, and \cite{CSS_Code_qudit} for detailed analysis of CSS qudit codes, EA qudit stabilizer, EA qudit CSS codes, EA qudit TPCs, and EA RS codes.}
\end{remark}

\subsection{TPC from RM, QRM {codes}, and entanglement assistance} \label{subsec:TPC_QRM_compare}

From the tensor product properties discussed in Section \ref{sec:prelim}, we now draw connections between the TPC constructed with RM codes and the QRM code.

\textit{QRM codespace as a subspace of TPC codespace:} 
The parity check matrix of the TPC constructed from classical codes $\RM(r_1, m_1)$ and $\RM(r_2, m_2)$ with parity check matrices $H_1 = \RM(m_1-r_1-1, m_1)$ and $H_2 = \RM(m_2 - r_2 -1, m_2)$ is given by
\begin{align}
    &H_1\otimes H_2\nonumber\\
    &= \RM(m_1 \!-\! r_1 \!-\!1, m_1)\otimes \RM(m_2\! -\!r_2 \!-\!1, m_2)\nonumber\\
    &\subseteq \RM(m_1 + m_2 - r_1 - r_2 - 2, m_1 + m_2),\\
    &~~~~~(\text{from equation } \eqref{eqn:tprm_in_rm})\nonumber
\end{align}
The QRM code $\mathrm{QRM}(r_1 + r_2 + 1, m_1 + m_2)$ has the parity check matrix $H = \RM(m_1 + m_2 - r_1 - r_2 - 2, m_1 + m_2)$. $H$ has the same minimum weight as the parity check matrix of the TPC and contains the parity checks of the TPC. Thus, for the considered case the QRM codespace forms a subspace of the TPC codespace. 

 For the parameter ranges $r_1 < (m_1 - 1)/2, r_2 < (m_2 - 1)/2$, the rows of $H_1\otimes H_2$ are not orthogonal and the quantum code construction requires entanglement assistance. In this range, the {quantum version of the RM} code also requires entanglement assistance. Thus, from Lemma \ref{lem:EARM_Zero_rate}, we find EA RM TPCs with positive catalytic rates, but we do not find corresponding {EA }RM with positive catalytic rates. The TPC construction using RM codes can be thought of as a systematic method to remove parity checks from the {EA} RM code. This removal leads to (i) a lower number of pre-shared entangled qubits required for orthogonalization, (ii) larger codespace, and thus larger coding rates. For example, for the parameters $r_1 = r_2 = 3$ and $m_1 = m_2 = 8$, we obtain the codes $\text{EA RM}(7, 16)$ with a catalytic rate $\mathcal R_C = -0.196$, while the code $\text{EA RM TPC}(\RM(3, 8), \RM(3, 8))$ has catalytic rate $\mathcal R_C = 0.189$. We note that if $r_1 \geq (m_1 - 1)/2$ or $r_2 \geq (m_2 - 1)/2$, the RM code is a dual-containing code and the RM TPC is a dual-containing TPC.

\textit{TPC codespace as a subspace of QRM codespace:}
Consider the codes $\mathrm{QRM}(r, m_1 + m_2)$ with the parity check matrix $H = \RM(m_1 + m_2 -r-1, {m_1+m_2})$ and TPC constructed with classical codes $\RM(0, m_1)$ and $\RM(r - m_1, m_2)$, $m_2 > m_1, r > m_1$ with the parity check matrix $H_1\otimes H_2 = \RM(m_1, m_1) \otimes \RM(m_1 + m_2 - r -1, m_2)$. By considering $r$ to be $m_1+m_2-r-1$ in equation \eqref{eqn:rm_in_tprm}, the parity checks of the two codes are related by
\begin{align}
    &\!\!\!\!\RM(m_1 \!+\! m_2 \!-\! r \!-\! 1, m_1\! +\! m_2) \nonumber\\ 
    &\!\!\!\!~~\subseteq \RM(m_1, m_1) \!\otimes\! \RM(m_1\! +\! m_2\! -\! r\! -\!1, m_2).\!\!
\end{align}
The parity checks of the TPC contain the parity checks of the QRM code. Thus, here the TPC codespace forms a subspace of the QRM codespace. For the parameter range $r < (m_1 + m_2 -1)/2$, the QRM and the TPC are not dual-containing and thus require entanglement assistance to construct quantum codes. In this parameter range with the additional constraint of $r < m_2$, while both have negative catalytic rates, we obtain EA {RM} TPCs with positive coding rate while the EA RM code has a zero rate. For parameters $m_1 = 5, m_2 = 11, r = 7$, we have the EA codes $\text{EA RM}(7, 16)$, and the EA {RM} TPC constructed with $\RM(0, 5)$ and $\RM(2, 11)$. The EA RM code has $\mathcal R_E = 0$, while the EA {RM} TPC constructed exhibits $\mathcal R_E = 0.002$, $\mathcal R_C = -0.874$ despite the parity check matrix of the TPC containing all the parity checks of the QRM code.

Thus, in both cases of containment, the EA {RM} TPC performs better than the corresponding EA RM code, illustrating that the benefit arises not simply from parameter ranges, but from structural differences in their construction.

\subsection{{Discussion on EA polar codes and their tensor product variants}}

{Polar codes are generalized variants of RM codes. A $2^m$-bit polar code that encodes $K$ message bits has a generator matrix comprising $K$ rows from the matrix $\begin{bmatrix}
    1 & 0\\1 & 1
\end{bmatrix}^{\otimes m}$, where the rows of the matrix chosen depend on the polarization of the channel under consideration. As the generator matrix of RM codes also comprises rows of $\begin{bmatrix}
    1 & 0\\1 & 1
\end{bmatrix}^{\otimes m}$, polar codes are {in some sense} generalized versions of }{RM codes.} 

{The polar code is obtained by considering a $2^m$-bit vector $u$ with $K$ locations containing the message bits and rest containing $0$. Let $f$ be a $2^m$-bit vector whose $i^{\mathrm{th}}$ element $f_i$ is $1$ if the $i^{\mathrm{th}}$ element of $u$ contains a message bit, else $f_i$ is $0$. The $K$ locations containing the message bits correspond to the index of the rows of $\begin{bmatrix}
    1 & 0\\1 & 1
\end{bmatrix}^{\otimes m}$ that belongs to the generator matrix $G$ of the polar code, i.e., $G = F\begin{bmatrix}
    1 & 0\\1 & 1
\end{bmatrix}^{\otimes m}$, where $F = \mathrm{diag}(f)$ is a diagonal matrix with the $i^{\mathrm{th}}$ diagonal element being $f_i$. The codeword for vector $u$ is $uG$. The polar code has parameters $[2^m, K, d]$, where the minimum distance $d$ depends on the form of $G$.}

{We next discuss the quantum version of a polar code based on a particular instance of polarization of the channels. As some quantum RM codes require entanglement-assistance, while others do not, quantum polar codes may or may not require entanglement-assistance depending on the generator matrix chosen. Let $H$ be the parity check matrix of the classical polar code $[2^m, K, d]$. The quantum polar code is a $[[2^m, 2K - 2^m + \mathrm{gfrank}(HH^{\mathrm{T}}), \geq d, \mathrm{gfrank}(HH^{\mathrm{T}})]]$ code.}

{When $\mathrm{gfrank}(HH^{\mathrm{T}}) = 2^m-2K$, the quantum polar code has zero rate. As the classical polar codes have non-zero rate, the quantum polar TPC constructed would have positive coding/entanglement-assisted rate.}
{A TPC constructed from two polar codes with parameters $[2^{m_1}, K_1, \geq d_1]$ and $[2^{m_2}, K_2, \geq d_2]$ has parameters $[2^{m_1+m_2}, K_12^{m_2} + K_22^{m_1} - K_1K_2, \mathrm{min}(d_1, d_2)]$. The EA polar TPC has parameters $[[2^{m_1+m_2},  K_12^{m_2+1} + K_22^{m_1+1} - 2^{m_1+m_2} - 2K_1K_2 + n_e, \geq \mathrm{min}(d_1, d_2); n_e = \mathrm{gfrank}(H_1H_1^{\mathrm{T}})\mathrm{gfrank}(H_2H_2^{\mathrm{T}})]]$. From Corollary \ref{coro:EACSS_positive_catalytic}, the catalytic rate is positive when $R_1 + R_2 > R_1R_2 + 0.5$ $\Rightarrow$ $\frac{K_1}{2^{m_1}} + \frac{K_2}{2^{m_2}} > \frac{K_1K_2}{2^{m_1+m_2}} + \frac{1}{2}$. When we consider both the component codes to be the same, we obtain the following:}
{
\begin{align}
     & \frac{K^2}{2^{2m}} + \frac{1}{2} < \frac{2K}{2^{m}} ,\nonumber\\
    \Rightarrow & 2K^2 - K2^{m+2} + 2^{2m} < 0,\nonumber\\
    \Rightarrow & 2(K^2 - 2K(2^m) + 2^{2m}) - 2^{2m} < 0,\nonumber\\
    \Rightarrow & 2(K - 2^{m})^2  < 2^{2m},\nonumber\\
    \Rightarrow & (2^{m} - K)^2  < \frac{2^{2m}}{2},\nonumber\\
    \Rightarrow & (2^{m} - K)  < \frac{2^{m}}{\sqrt{2}},\nonumber\\
    \Rightarrow & (2^{m} - \frac{2^{m}}{\sqrt{2}})  < K, \nonumber\\
    \Rightarrow & 2^{m}\frac{\sqrt{2} - 1}{\sqrt{2}} = \frac{2^{m}}{2 + \sqrt{2}}  < K.
 \end{align}}

 {We note that this is the exact relation we obtained for RM codes, where $K$ was a function of $l(r)$. We obtain a natural boundary here as well.}



\section{Conclusion{s and Open Problems}}\label{sec:Conclusion}
We showed that the EA RM codes constructed from classical binary RM codes that are not dual-containing codes have zero coding rates and negative catalytic rates. Utilizing the {EA} tensor product code construction with the same classical binary RM codes, we constructed EA RM TPCs that have positive coding rates; thereby, demonstrating the coding analog of superadditivity. We also showed that some of these codes, such as EA RM TPCs constructed from $\mathrm{RM}(i, 2i+2)$ codes, for $i \in \mathbb{Z}^+$, can have positive catalytic rate, useful for building robust quantum communication systems. 

{We showed that the quantum TPC construction can be used to obtain positive rate EA polar TPCs whose component codes yield zero rate EA polar codes. We provided the condition for these EA polar TPCs to have positive catalytic rates. We also showed that positive rate EA RM TPCs over qudits can be constructed for any RM code. A rigorous analysis of simplified conditions for EA RM TPCs over qudits and EA polar TPCs to have positive catalytic rates and deriving their relations to larger length EA RM qudit codes and EA polar codes are future directions of research.}

{Punctured quantum Reed Muller codes have been shown to admit non-Clifford transversal gate operations and are useful for fault tolerant quantum computation. However, the EA variants do not necessarily admit transversal gate operations. Finding the parameter ranges and conditions under which transversal gates are possible for an EA code and determining the EA construction to achieve this is an interesting open problem.}

{One of the important quests in quantum information theory is to identify the necessary and sufficient conditions for channels to be superadditive under a generic setting. This is an open problem. With zero-rate channel capacities, one can make use of entanglement-assistance using the tensor product coding framework using codes with suitable geometries to achieve positive catalytic rates.}  

{A natural extension of this work is to generalize the results in this paper for establishing coding theorems over entanglement-assisted networks within a graph-theoretic setup. These are some of the interesting directions to pursue towards quantum network information and coding theories.}

\section*{Code availability}
Numerical results in Tables \ref{tab:l_r} and \ref{tab:examples} and the parameter boundaries depicted in Figure \ref{fig:m-r-plane} were obtained using code developed from standard Python libraries and can be found at \url{https://github.com/Praveen91299/EA_RM_TPC.git}

\section*{Acknowledgment}
P. Jayakumar and A. Behera acknowledge fellowship support from Kishore Vaigyanik Protsahan Yojana (KVPY) scheme, Government of India, for their undergraduate research.

\onecolumn
\newpage
\appendix

\section{Properties of evaluation vector function $\mathrm{Eval}^{(m)}(\cdot)$}\label{app:prop}
The following lemmas hold true for any finite field $\mathbb{F}_d, d\in \mathbb{Z}_+$. We restrict to $\mathbb{F}_2$ throughout this paper.

\begin{replemma}{lemma:linearity}[Linearity] 
    For $f, g : \mathbb{F}_2^m \rightarrow \mathbb{F}_2; f, g \in W_m := \mathbb{F}_2[x_1, x_2, \cdots, x_m]$ the evaluation function has the following properties:
\begin{subequations}
    \begin{align}
    (i)~~& \mathrm{Eval}^{(m)}(f + g) = \mathrm{Eval}^{(m)}(f) \oplus \mathrm{Eval}^{(m)}(g)\label{prop:eval1app}\\
    (ii)~~& \mathrm{Eval}^{(m)}(fg) = \mathrm{Eval}^{(m)}(f)\odot\mathrm{Eval}^{(m)}(g)\label{prop:eval2app}
\end{align}
\end{subequations}
where $\odot$ denotes the bitwise product of the two vectors {$\odot : \mathbb{F}_2^{2^m} \times \mathbb{F}_2^{2^m} \rightarrow \mathbb{F}_2^{2^m}$.}
\end{replemma}
\begin{proof}
    Using linearity of the $\mathrm{Eval}_{\bm{z}}(\cdot)$ function we first prove equation \eqref{prop:eval1app}, and from there we proceed to prove equation \eqref{prop:eval2app}. For a point $\bm{z} = (z_1, z_2, \cdots, z_m)$ and two linear polynomial functions $f, g \in W_m$, $\mathrm{Eval}_{\bm{z}}(f+g) = (f+g)(\bm{z}) = f(\bm{z}) + g(\bm{z})$. But $f(\bm{z}) = \mathrm{Eval}_{\bm{z}}(f)$ and $g(\bm{z}) = \mathrm{Eval}_{\bm{z}}(g)$. We have $\mathrm{Eval}_{\bm{z}}(f+g) = \mathrm{Eval}_{\bm{z}}(f) + \mathrm{Eval}_{\bm{z}}(g)$ $\forall \bm{z} \in \mathbb{F}_2^m$, thus equation \eqref{prop:eval1app} follows.
    
    To prove equation \eqref{prop:eval2app}, we first show that it is true when $f, g$ are monomials and then generalize to polynomials using equation \eqref{prop:eval1app}. For index sets A and B, consider $f = x_A$, $g = x_B$, then $fg = x_C$ where $C = A \cup B$ (as if $i \in A$ or $i \in B$, then $i \in C$ as $x_i^j = x_i$ for $j\geq 1, j\in \mathbb{Z}^+$). Recall that $x_A = \prod_{i\in A}x_i$ for an index set $A$. If $i \in A$ ($i \in B$), when $z_i = 0$ we have $\mathrm{Eval}_{\bm{z}}(x_A) = 0$ ($\mathrm{Eval}_{\bm{z}}(x_B) = 0$) and also $\mathrm{Eval}_{\bm{z}}(x_C) = 0$. Whenever $\mathrm{Eval}_{\bm{z}}(x_A) = 0$ or $\mathrm{Eval}_{\bm{z}}(x_B) = 0$, $\mathrm{Eval}_{\bm{z}}(x_C) = \mathrm{Eval}_{\bm{z}}(fg) = 0$. Thus we conclude that equation \eqref{prop:eval2app} is true for monomials. Now consider $f = \sum_{i=1}^m a_ix_{A_i}$ and $g = \sum_{i=1}^n b_ix_{B_i}$ where $A_i, B_j \in \{1, 2,\cdots, m\} \forall i, j$, then $fg = \sum_{i, j = 1}^{m, n} a_ib_jx_{A_i}x_{B_j}$. Using equation \eqref{prop:eval1app}, we have 
    \begin{align}
        &\mathrm{Eval}^{(m)}(fg)= \mathop{\bigoplus}_{i, j = 1}^{m, n} a_i b_j \mathrm{Eval}^{(m)}(x_{A_i}x_{B_j})\nonumber\\
        &= \mathop{\bigoplus}_{i, j = 1}^{m ,n} a_i b_j \mathrm{Eval}^{(m)}(x_{A_i})\odot\mathrm{Eval}^{(m)}(x_{B_j})\nonumber\\
        &= \mathop{\bigoplus}_{i= 1}^{m} \left(a_i \mathrm{Eval}^{(m)}(x_{A_i})\odot\left(\mathop{\bigoplus}_{j = 1}^{n} b_j\mathrm{Eval}^{(m)}(x_{B_j})\right)\right)\nonumber
        \end{align}
        \begin{align}
        &= \left(\mathop{\bigoplus}_{i= 1}^{m} a_i \mathrm{Eval}^{(m)}(x_{A_i})\right)\odot\left(\mathop{\bigoplus}_{j = 1}^{n} b_j\mathrm{Eval}^{(m)}(x_{B_j})\right)\nonumber\\
        &= \mathrm{Eval}^{(m)}(f)\odot\mathrm{Eval}^{(m)}(g)\qedhere
    \end{align}
\end{proof}

\begin{replemma}{lemma:tpRM}[Tensor product]
    The tensor product of evaluation vectors of polynomials over the field $\mathbb{F}_2$ is an evaluation vector of a polynomial from a larger ring of the same field.
    In particular, for positive integers $m_1, m_2$ and $i \in  [m_1], j \in [m_2]$ where $[m] = \{1, 2, \cdots, m\}$ the evaluation vector has the property:
    \begin{align}
        (i) & \mathrm{Eval}^{(m_1)}(x_i) \!\otimes\! \mathrm{Eval}^{(m_2)}(x_j)\! =\! \mathrm{Eval}^{(m_1 + m_2)}(x_ix_{m_1 + j})\label{tp1app}\\
        (ii) & \mathrm{Eval}^{(m_1)}(x_A)\! \otimes \! \mathrm{Eval}^{(m_2)}(x_B)\! =\! \mathrm{Eval}^{(m_1 + m_2)}(x_Ax_{m_1 + B})\label{tp2app}
    \end{align}
    where $A \subseteq [m_1], B \subseteq [m_2]$ are index sets and $m_1 + B = \{m_1 + b | b \in B\}$ is a shift of indices in set $B$.
\end{replemma}
\begin{proof}
    $\mathrm{Eval}^{(m_1)}(x_i)$ and $\mathrm{Eval}^{(m_2)}(x_j)$ are binary strings of length $2^{m_1}$ and $2^{m_2}$. Thus, the tensor product has a length of $2^{m_1 + m_2}$. For $\bm z \in \mathbb{F}_2^{m_1 + m_2}$, if $z = z^{(1)}2^{m_2} + z^{(2)} \equiv (z^{(1)}, z^{(2)})$ is the $2^{m_2}$-ary representation of $z$ with $0\leq z^{(1)} < 2^{m_1}, 0\leq z^{(2)}< 2^{m_2}$. Further, if $\bm z = (z_1, z_2, \cdots, z_{m_1}, z_{m_1 + 1}, \cdots, z_{m_1 + m_2})$ is the binary representation of $z$, then $\bm z^{(1)} = (z_1, z_2, \cdots, z_{m_1})$ and $\bm z^{(2)} = (z_{m_1 + 1}, z_{m_1 + 2}, \cdots, z_{m_1 + m_2})$ are the binary representations of $z^{(1)}, z^{(2)}$ respectively.

    The $z^{\mathrm{th}}$ bit in $\mathrm{Eval}^{(m_1 + m_2)}(x_ix_{m_1 + j})$ is $\mathrm{Eval}_{\bm z}(x_ix_{m_1 + j}) =  z_iz_{m_1 + j}$. If $z = z^{(1)}2^{m_2} + z^{(2)} \equiv (z^{(1)}, z^{(2)})$ is the $2^{m_2}$-ary representation of $z$, from the definition of tensor product, the $z^{\mathrm{th}}$ bit in $\mathrm{Eval}^{(m_1)}(x_i) \otimes \mathrm{Eval}^{(m_2)}(x_j)$ is given by $z_i^{(1)}z_j^{(2)}$. But $z_i^{
(1)} = z_i$ and $z_j^{(2)} = z_{m_1 + j}$. Thus, equation \eqref{tp1app} holds true. Equation \eqref{tp2app} follows trivially using equation \eqref{prop:eval2} and the distributivity property of tensor products $(a_1\otimes b_1)\cdot(a_2\otimes b_2) = (a_1a_2\otimes b_1b_2)$.
\end{proof}
\section{Product codes and tensor product codes}\label{app:product}
The product code obtained from two codes is equivalent to a subset of the tensor product code obtained from the same codes. A product code is a 2-D linear code wherein the rows are encoded with one code and the columns with another code. For the $[n_1, k_1, d_1]$ and $[n_2, k_2, d_2]$ classical codes with generator matrices $G_1, G_2$, the state encoded with the product code is given by $E = G_2^{\mathrm{T}}MG_1$, where $M$ is the $k_2\times k_1$ message block. Denote the unit vectors of length $k$ as $\ket{i^{(k)}} = (0, 0, \cdots, 1, 0, \cdots, 0)$ with $1$ in the $i^{\mathrm{th}}$ position. Any $k_2 \times k_1$ matrix $M$ can be written as $M = \sum_{i = 0, j = 0}^{k_2 -1, k_1 - 1}M_{i, j}\ket{i}\bra{j}$. Similarly $E = \sum_{i = 0, j = 0}^{k_2 -1, k_1 - 1}M_{i, j}G_2^{\mathrm{T}}\ket{i}\bra{j}G_1$. One can vectorize the matrix as $M\equiv \sum_{i = 0, j = 0}^{k_2 -1, k_1 - 1}M_{i, j}\bra{j}(\ket{i})^\T = \sum_{i = 0, j = 0}^{k_2 -1, k_1 - 1}M_{i, j}\bra{j}\bra{i}$ which is equivalent to appending the rows of $M$ into a single vector. Vectorizing the encoded state, 
\begin{align}
    E &\equiv\!\!\!\!\! \sum_{i = 0, j = 0}^{k_2 -1, k_1 - 1}\!\!\!\!\!\!\!\!M_{i, j}(\bra{j}G_1)(G_2^\T\ket{i})^\T= \!\!\!\sum_{i = 0, j = 0}^{k_2 -1, k_1 - 1}\!\!\!\!\!\!\!\!M_{i, j}(\bra{j}G_1)(\bra{i}G_2)= \sum_{i = 0, j = 0}^{k_2 -1, k_1 - 1}M_{i, j}\bra{j}\bra{i}(G_1\otimes G_2).
\end{align}
We obtain a 1-D linear code with the generator matrix given by $G_1\otimes G_2$.
For an RM code, since the tensor product of generator matrices of $\RM(r_1, m_1)$ and $\RM(r_2, m_2)$ form a subset of the rows of the generator matrix of $\RM(r_1 + r_2, m_1 + m_2)$, the product code obtained from $\RM(r_1,m_1)$ and $\RM(r_2, m_2)$ is equivalent to a subset of the $\RM(r_1 + r_2, m_1 + m_2)$ code.
\section{Proofs of Lemmas on coding and catalytic rates of EA RM codes and EA RM TPCs}\label{app:EARM_rate_Proof}
\begin{reptheorem}{thm:EARMTPC_PositiveCodingRate}
The EA RM TPC obtained from the $\mathrm{RM}(r_1,m_1)$ and $\mathrm{RM}(r_2,m_2)$ classical RM codes have positive coding rate.
\end{reptheorem}
\begin{proof}
The EA RM TPC obtained from two classical RM codes $\mathrm{RM}(r_1,m_1)$ and $\mathrm{RM}(r_2,m_2)$, namely, $C_1[2^{m_1}, \underset{i=0}{\overset{r_1}{\sum}} {m_1 \choose i}, 2^{m_1-r_1}]$ and $C_2[2^{m_2}, \underset{i=0}{\overset{r_2}{\sum}} {m_2 \choose i}, 2^{m_2-r_2}]$ has the parameters $[[n, n-2\rho + n_e, \geq d{;} n_e]]$, where $n_e = \mathrm{gfrank}(HH^{\mathrm{T}}) = \mathrm{gfrank}(H_1H_1^{\mathrm{T}})\mathrm{gfrank}(H_2H_2^{\mathrm{T}})$ \cite{EABQTPC}, $\rho = \rho_1\rho_2 = \left(\underset{i=0}{\overset{m_1-r_1-1}{\sum}} {m_1 \choose i}\right)\left(\underset{i=0}{\overset{m_2-r_2-1}{\sum}} {m_2 \choose i}\right)$ and $n_e = n_{e1}n_{e2} = \left(\underset{i=r_1+1}{\overset{m_1-r_1-1}{\sum}} {m_1 \choose i}\right)\left(\underset{i=r_2+1}{\overset{m_2-r_2-1}{\sum}} {m_2 \choose i}\right)$. Thus, the number of logical qubits is 
\begin{align}
&n-2\rho+n_e\nonumber\\
& = 2^{m_1+m_2} - 2\left(\underset{i=0}{\overset{m_1-r_1-1}{\sum}} {m_1 \choose i}\right)\left(\underset{i=0}{\overset{m_2-r_2-1}{\sum}} {m_2 \choose i}\right) + \left(\underset{i=r_1+1}{\overset{m_1-r_1-1}{\sum}} {m_1 \choose i}\right)\left(\underset{i=r_2+1}{\overset{m_2-r_2-1}{\sum}} {m_2 \choose i}\right),\nonumber\\
&= 2^{m_1+m_2} \!-\! \left(\underset{i=0}{\overset{m_1-r_1-1}{\sum}} {m_1 \choose i}\!\right)\!\left(\underset{i=0}{\overset{m_2-r_2-1}{\sum}} {m_2 \choose i}\!\right)\!-\! \left(\underset{i=0}{\overset{r_1}{\sum}}\! {m_1 \choose i}\!\! +\!\! \underset{i=r_1+1}{\overset{m_1-r_1-1}{\sum}} \!\!{m_1 \choose i}\!\!\right)\!\!\left(\!\underset{i=0}{\overset{r_2}{\sum}}\! {m_2 \choose i}\!\! +\!\!\! \underset{i=r_2+1}{\overset{m_2-r_2-1}{\sum}} \!{m_2 \choose i}\!\right)\nonumber\\
&\hspace{9.5cm}+ \left(\underset{i=r_1+1}{\overset{m_1-r_1-1}{\sum}} {m_1 \choose i}\right)\left(\underset{i=r_2+1}{\overset{m_2-r_2-1}{\sum}} {m_2 \choose i}\right),\nonumber\\
&= 2^{m_1+m_2} - \left(\underset{i=0}{\overset{m_1-r_1-1}{\sum}} {m_1 \choose i}\right)\left(\underset{i=0}{\overset{m_2-r_2-1}{\sum}} {m_2 \choose i}\right)- \!\left(\underset{i=0}{\overset{r_1}{\sum}} \!{m_1 \choose i}\!\!\right)\!\!\left(\underset{i=0}{\overset{r_2}{\sum}}\! {m_2 \choose i}\!\!\right)\!\!\nonumber\\
&~~~ - \left(\underset{i=0}{\overset{r_1}{\sum}} {m_1 \choose i}\!\!\right)\!\!\left(\underset{i=r_2+1}{\overset{m_2-r_2-1}{\sum}} {m_2 \choose i}\!\!\right) - \left(\underset{i=r_1+1}{\overset{m_1-r_1-1}{\sum}} {m_1 \choose i}\right)\left(\underset{i=0}{\overset{r_2}{\sum}} {m_2 \choose i}\right)\nonumber\\
&~~~- \left(\underset{i=r_1+1}{\overset{m_1-r_1-1}{\sum}} {m_1 \choose i}\right)\left(\underset{i=r_2+1}{\overset{m_2-r_2-1}{\sum}} {m_2 \choose i}\right)+ \left(\underset{i=r_1+1}{\overset{m_1-r_1-1}{\sum}} {m_1 \choose i}\!\right)\!\!\left(\underset{i=r_2+1}{\overset{m_2-r_2-1}{\sum}} {m_2 \choose i}\!\right),\nonumber\\
&= 2^{m_1+m_2} - \left(\underset{i=0}{\overset{m_1-r_1-1}{\sum}} {m_1 \choose i}\right)\left(\underset{i=0}{\overset{m_2-r_2-1}{\sum}} {m_2 \choose i}\right)- \left(\underset{i=0}{\overset{m_1-r_1-1}{\sum}} {m_1 \choose i}\right)\left(\underset{j=m_2-r_2}{\overset{m_2}{\sum}} {m_2 \choose m_2-j}\right)\nonumber\\
&\hspace{7.5cm}-\left(\underset{i=0}{\overset{r_1}{\sum}} {m_1 \choose i}\!\right)\!\!\left(\underset{i=r_2+1}{\overset{m_2-r_2-1}{\sum}} {m_2 \choose i}\!\right)\!,\text{ where }j=m_2-i\nonumber\\
&=\! \left(\underset{i=0}{\overset{m_1}{\sum}} {m_1 \choose i}\!\right)\!\!\left(\underset{i=0}{\overset{m_2}{\sum}} \!{m_2 \choose i}\!\right)\!\! -\!\! \left(\underset{i=0}{\overset{m_1-r_1-1}{\sum}}\! {m_1 \choose i}\!\right)\!\!\left(\underset{i=0}{\overset{m_2}{\sum}} {m_2 \choose i}\!\!\right)-\left(\underset{i=0}{\overset{r_1}{\sum}} {m_1 \choose i}\right)\left(\underset{i=r_2+1}{\overset{m_2-r_2-1}{\sum}} {m_2 \choose i}\right)\nonumber\\
&\hspace{6cm}\left(\because 2^m = \left(\underset{i=0}{\overset{m}{\sum}} {m \choose i}\right)\text{ and } {m_2 \choose m_2-j} = {m_2 \choose j}\right)\nonumber\\
&= \left(\underset{i=m_1-r_1}{\overset{m_1}{\sum}} {m_1 \choose i}\right)\left(\underset{i=0}{\overset{m_2}{\sum}} {m_2 \choose i}\right)-\left(\underset{i=0}{\overset{r_1}{\sum}} {m_1 \choose i}\right)\left(\underset{i=r_2+1}{\overset{m_2-r_2-1}{\sum}} {m_2 \choose i}\right), \nonumber\\
&= \left(\underset{j=0}{\overset{r_1}{\sum}} {m_1 \choose j}\right)\left(\underset{i=0}{\overset{m_2}{\sum}} {m_2 \choose i}-\underset{i=r_2+1}{\overset{m_2-r_2-1}{\sum}} {m_2 \choose i}\right), \text{ where }j=m_1-i,~~~\left(\because {m_1 \choose j} = {m_1 \choose m_1-j}\right)  \label{eqn:EARMTPC_logical_qubits_no}\\
&= \left(\underset{j=0}{\overset{r_1}{\sum}} {m_1 \choose j}\right)\left(\underset{i=m_2-r_2}{\overset{m_2}{\sum}} {m_2 \choose i}+\underset{i=0}{\overset{r_2}{\sum}} {m_2 \choose i}\right)\nonumber\\
&>0 ~~\left(\because ~r_1,r_2 \geq 0\text{ and }m_2 \geq m_2-r_2\right). \nonumber
\end{align}
Thus, the number of logical qubits is non-zero and the coding rate is non-zero.
\end{proof}
\vspace*{2\baselineskip}

\end{document}